\documentclass[11pt]{amsart}
\usepackage{amsmath, amssymb, amsthm, mathrsfs, verbatim, epsf, graphicx, multirow, color, fullpage}

\usepackage[all]{xy}
\xyoption{poly}
\xyoption{arc}
  \usepackage{geometry}
\usepackage[colorinlistoftodos,textsize=footnotesize]{todonotes}
\usepackage{soul}
\setstcolor{red}

\newcommand{\Z}{\mathbb{Z}}
\newcommand{\C}{\mathbb{C}}
\newcommand{\R}{\mathbb{R}}
\newcommand{\N}{\mathbb{N}}

\newcommand{\D}{\partial}
\newcommand{\rk}{\textrm{rk }}
\newcommand{\spn}{\textrm{span }}

\def\ba #1\ea{\begin{align} #1 \end{align}}
\def\bas #1\eas{\begin{align*} #1 \end{align*}}
\def\bml #1\eml{\begin{multline} #1 \end{multline}}
\def\bmls #1\emls{\begin{multline*} #1 \end{multline*}}

\newcommand{\fB}{\mathfrak{B}}
\newcommand{\fP}{\mathcal{P}}

\newcommand{\cM}{\mathcal{M}}
\newcommand{\cA}{\mathcal{A}}
\newcommand{\cC}{\mathcal{C}}
\newcommand{\Prop}{\textrm{Prop}}

\newcommand{\rd}{\textrm{d}}
\newcommand{\cZ}{\mathcal{Z}}
\newcommand{\Dom}{\textrm{Dom}}
\newcommand{\detzr}[1] {\langle (\cZ_*^\mu|V_p)^{#1} \rangle}
\newtheorem{thm}{Theorem}[section]

\newtheorem{lem}[thm]{Lemma}
\newtheorem{cor}[thm]{Corollary}
\newtheorem{prop}[thm]{Proposition}

\theoremstyle{remark}
\newtheorem{eg}[thm]{Example}

\theoremstyle{definition}
\newtheorem{dfn}[thm]{Definition}
\newtheorem{rmk}[thm]{Remark}

\begin{document}
\title{Wilson Loop diagrams and Positroids}
\author{Susama Agarwala}
\author{Eloi Marin Amat}
\date{\today}

\begin{abstract}
In this paper, we study a new application of the positive Grassmanian to  Wilson loop diagrams (or MHV diagrams) for scattering amplitudes in N=4 Super Yang-Mill theory ($N=4$ SYM).  There has been much interest in studying this theory via the positive Grassmanians using BCFW recursion. This is the first attempt to study MHV diagrams for planar Wilson loop calculations (or planar amplitudes) in terms of positive Grassmannians. We codify Wilson loop diagrams
completely in terms of matroids. This allows us to apply the combinatorial tools in matroid theory used to identify positroids, (non-negative Grassmannians), to Wilson loop diagrams. In doing so, we find that certain non-planar Wilson loop diagrams define positive Grassmannians. While non-planar diagrams do not have physical meaning, this finding suggests that they may have value as an algebraic tool, and deserve further investigation.
\end{abstract}

\maketitle

\tableofcontents

\makeatletter
\providecommand\@dotsep{5}
\makeatother

During the last decade, the computation of scattering amplitudes in N=4 SYM has evolved away from old-school Feynman diagrams to the use of twistors and recursive methods that are much more efficient computationally. The two most well-known methods are the  BCFW recursion relations \cite{Britto:2005fq} and MHV diagrams as introduced by Cacazo, Svrcek and Witten \cite{Cachazo:2004kj} and developed using the twistor action \cite{Adamo:2011cb}. Both methods produce a sum of terms.  The BCFW recursion relations express amplitudes as sums of terms built from amplitudes involving a smaller number of particles that can eventually be built from three-point amplitudes.  MHV diagrams express amplitudes as a sum of terms that have a representation as diagrams.  Recently, a map between contributions arising from the BCFW recursion relations and certain cycles in the positive Grassmanian have been shown to play a key role in the theory \cite{arkani:2012nw}. In a further development, these cycles can be pieced together inside another Grassmanian, to form a `polyhedron' that has been christened the  Amplituhedron \cite{Arkani-Hamed:2013jha}. In this paper we pursue just the first step in creating a parallel construction for MHV diagrams and study the correspondence between cycles in the positive grassmanian and MHV diagrams.

Scattering amplitudes for $N=4$ SYM can be obtained from a formulation of the theory via an action in twistor space that yields MHV diagrams as the Feynman diagrams \cite{Boels:2006ir,Boels:2007qn}. In a separate development, it emerges that amplitudes can be obtained  in a completely different way as certain correlation functions of Wilson loops in planar $N=4$ SYM \cite{Alday:2007hr}.  A polygon can be obtained from an amplitude by joining together the momentum vectors of the particles   taking part in the interaction process (the planarity assumption gives an ordering to the particles). The sides are null vectors so in space-time this polygon is constrained, but it can be reformulated as a generic polygon in twistor space which is complex projective 3-space $\mathbb{CP}^3$.  In order to take account of the different polarization states of particles, a supersymmetric formulation is often used in which twistor space acquire further fermionic coordinates.  We adopt a formalism developed in \cite{hodges:2013eliminating, Arkani-Hamed:2013jha} that expresses these in terms of $k$ additional bosonic coordinate that encode this extra fermionic dependence.  Here $k$ denotes the MHV degree for the amplitude that corresponds to there being  $k+2$  negative helicity  gluons in the interaction (the amplitude vanishes when $k=-2$ or $-1$ hence the maximality).   A key advantage of the MHV diagram formalism is that $k$ also counts the number of internal edges in a diagram and so amplitudes with low $k$ are very simple.

Thus the basic data on which an amplitude or Wilson-loop correlator depends, is $n$ twistors with values in $\C^{4+k}$ or projectively in $\mathbb{CP}^{3+k}$.   A tree amplitude, or more generally the integrand for a loop amplitude is a rational function of this  data.   In this paper we will be concerned with the representation of this rational function as a sum of contributions coming from MHV diagrams (as opposed to BCFW recursion which leads to so-called on-shell diagrams or even traditional Feynman diagrams obtained from a space-time rather than twistor space action).  We will see that the formula for these MHV diagrams  involves an integral over $4k$ parameters that have a natural interpretation as coordinates on a $4k$ cycle in the Grassmannian $G(k,n)$.

Because we are dealing with rational functions, we will be flexible as to whether the basic variables are taken to be real or complex, as the complexification will be unambiguous. Indeed we see that we have important additional structure if, rather than $\C^{4+k}$, we take the data to be in $\R^{4+k}$. In this case the Grassmannian $G(k,n)$ has a subspace $G_+(k,n)$, the positive Grassmannian, on which cyclic minors are positive. The intersection of the $4k$ cycles with the boundary of this space encodes the poles of the rational functions.  These poles fall into two classes, the spurious poles, which identify both the sums of diagrams that lead to the final, physical, amplitude, and the the physical poles that are an essential ingredient in the final amplitude.

In this paper, we follow in the footsteps of those, such as \cite{arkani:2012nw}, in that we use the language of positive Grassmannians and matroids to understand amplitudes arising in $N=4$ SYM theory. However, the similarity ends there. This paper focuses on MHV digrams, which allows one to study the entire family of off shell interactions missed by the BCFW approach. However, unlike the BCFW approach, there is not a relationship between the Amplitudes that arise in this context and the cycles in the positive Grassmannians that they diagrams naturally define. In particular, we find that for a particular class of Wilson loop diagrams, those with $\overline{MHV}$ subdiagrams, there is an ambiguity in the correspondence between the Amplitudes and Wilson loop diagrams, Theorem \ref{equivalencethm} and Remark \ref{amplitudeambiguity}. While it is well known that all $\overline{MHV}$ Wilson loop amplitudes (determined by sums of $\overline{MHV}$ Wilson loop diagram amplitudes) are trivial, there is no such result for the large class of diagrams the contain $\overline{MHV}$ \emph{sub}diagrams. Furthermore, we find, exactly in this case of Wilson loop diagrams with $\overline{MHV}$ subdiagrams, that non-planar Wilson loop diagrams also give rise to positive Grassmannians, Theorem \ref{crosspos}. While there are reasons to believe that non-planar diagrams are not physically significant for the understanding of this theory, initial calculations show that they may be important, if only as algebraic tools, for understanding the relationship between the Grassmanians that arise in the MHV context, and the Amplituhedron defined in the BCFW context. Whatever the case, this paper shows that the class of Wilson loop diagrams with $\overline{MHV}$ subdiagrams deserve significant further study.

There are other advantages to studying Wilson loop diagrams using matroids. All of the key properties defining a matroid are encoded pictorially in the Wilson loop diagram. Therefore, by undertanding both matroids and Wilson loop diagrams, one may literally read off properties of the associated Grassmannian from the picture. This powerful innovation allows one to prove theorems at all $N^kMHV$ levels, where as direct calculational methods are still restricted to the case of $k=2$. Using these techniques, we have proven the long standing adage that planar Wilson loops lead to positive Grassmannians, Corollary \ref{noncrosspos} and Theorem \ref{noncrosspartconn}. As mentioned above, this work also leads to the intriguing realization certain diagrams with crossing partitions also yeild positive Grassmannians, Theorem \ref{crosspos}. More concretely, the matroidal approach finds a way to decompose Wilson loop diagrams into indecomposable building blocks, consisting of fewer vertices and propagators, which will help in future calculations of complicated diagrams at higher MHV degree, Theorem \ref{disconnectedloopprop}.

The rest of this paper is organized as follows. Section \ref{physics} introduces the Wilson loop diagrams and the amplitudes they represent. It explains the connection between Wilson loop diagrams and the Grassmanians they define. This section also identifies the aforementioned class of Wilson loop diagrams with $\overline{MHV}$ subdiagrams, and give conditions for when these diagrams correspond to the same matroid, Theorem \ref{equivalencethm}. Finally, Theorem \ref{overdefinethm} show that any Wilson loop diagram containing a subdiagram at $n\leq k+2$ is trivial. Section \ref{matroiddefssection} introduces the key concepts from matroid theory that are crucial for this paper. This includes a criteria for understanding positivity in the language of matroids. Section \ref{Wilsonloopstopositroids} contains the key results of this paper. In it, we use the notation set up in section \ref{matroiddefssection} to define, graphically, which Wilson loop diagrams satisfy the positivity conditions outlined in section \ref{physics}. In doing so, we show how to decompose a Wilson loop diagram into its component building blocks Theorem \ref{disconnectedloopprop}. We show that planar Wilson loops lead to positive Grassmannians, Corollary \ref{noncrosspos} and Theorem \ref{noncrosspartconn}, and that certain diagrams with crossing partitions also yeild positive Grassmannians, Theorem \ref{crosspos}. Section \ref{futurework} discusses future directions for this work.

\section{Wilson loop diagrams and amplitudes \label{physics}}

This section is an introduction to the combinatorics of Wilson loop diagrams for mathematicians. As such, we leave the precise definitions of the physical objects involved, and the derivation of the diagrams and associated integrals to the existing literature, and focus what the defined amplitudes and the associated diagrams at the heart of this paper are.

\subsection{The Feynman rules of Wilson loops} \label{Feynmanrules}
As remarked above, amplitudes and Wilson loops in planar $N=4$ Super-Yang-Mills are dual (equivalent) objects.  They are given as sums of Feynman diagrams and we focus on those that arise from the Wilson-loop description, but reformulated as holomorphic Wilson-loops in twistor space
\cite{Mason:2010yk}. For more on the relation between Wilson loop diagrams and traditional Feynman diagrams, see \cite{Adamo:2011cb}.

For a tree level amplitude (no internal loops), a Wilson loop diagram consists of a boundary polygon, with $n$ cyclically ordered vertices and $k$ propagators, with $n \geq 4+k$. The planarity assumption in Yang-Mills arises from a limit in which a gauge group $SU(n)$ is chosen in which we take $n\rightarrow\infty$ and in that limit the leading contributions are the planar diagrams\footnote{although a non-planar Wilson-loop does make sense, non-planar diagrams come with different colour factors that encode the gauge group dependence of the particles.}.  We may however consider non-planar diagrams for mathematical reasons.  Propagators of the interaction are represented by wavy lines in the diagram connecting two sides of a Wilson loop diagram. Each propagator is defined by an ordered pair $(i, j)$, where $i, i+1$ and $j, j+1$ define the two edges of the boundary polygon. A key simplification of the MHV diagram formalism is that there are no additional vertices at the tree level.  For an $L$-loop amplitude, there are $L$ vertices in the holomorphic Wilson-loop diagram, and each vertex can be represented geometrically as a further line in twistor space that generically does not meet the polygon or other such lines.   We will ignore these for the most part and focus on tree level diagrams in this paper. The key ideas extend to that case quite simply.

\begin{dfn}
A Wilson loop diagram is comprised of a cyclic ordered set $[n]$, and a set of $k$ pairs of its elements: \bas \fP = \{(i_1, j_1), \ldots (i_k, j_k) | i_r, j_r \in [n]; \leq n\}\;. \eas Then a, Wilson loop diagram, $W$, is the pair $W = (\fP, [n])$.
\end{dfn}

Note that, for the moment, we make no requirement that $i_p < j_p$. Therefore, $(i, j)$ and $(j , i)$ represent the same propagator.

In the original Wilson loop diagram, the vertices of the boundary polygon correspond to complex super-twistors, but here we follow the amplituhedron convention \cite{Arkani-Hamed:2013jha} that eliminates the four fermionic variables in favour of $k$ bosonic ones and takes the twistors to be real so that the vertices $Z_1 \ldots Z_n \in \R^{4+k}$.   We write each $Z = (Z^\mu, \vec{z}) \in \R^{4+k}$, with $Z^\mu \in \R^4$ encoding the momentum of the particles, and $\vec{z} \in \R^k$, the bosonized fermions. The set of vectors $\{Z_1 \ldots Z_n\}$ are chosen such that any ordered subset of $4+k$ twistors defines a positive volume.

\begin{dfn} Indicate by $\cZ \in M(n, 4+k)$, the matrix with rows defined by $Z_1 \ldots Z_n$ in that order, all maximal determinants of $\cZ$ are positive. \end{dfn}

A Wilson loop diagram contributes to an $N^kMHV$ at $n$ points if there are $k$ internal propagators and $n$ external boundary vertices.  It also depends on a reference twistor denoted $Z_*$. The dependence on $Z_*$ only drops out in the final sum over all diagrams.

\begin{eg}
Here is an 8 point Wilson loop diagram with one MHV propagator, i.e., NMHV:
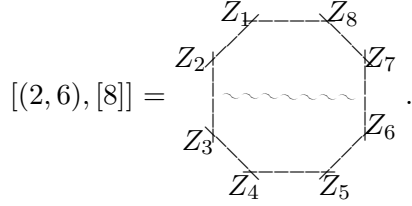
\bas \left[(2, 6), [8] \right] = {\begin{xy}
(-6, 10); (6, 10) **{\dir{-}}, (7, 11) *{Z_8},
(-4, 11); (-11, 4) **{\dir{-}}, (-13, 5) *{Z_2},
 (-10, 6); (-10, -6) **{\dir{-}}, (-12, -6) *{Z_3},
(-11, -4); (-4, -11) **{\dir{-}}, (-6, -12) *{Z_4},
(-6, -10); (6, -10)  **{\dir{-}}, (6, -12) *{Z_5},
(4, 11); (11, 4) **{\dir{-}}, (-7,11) *{Z_1},
(10, -6); (10, 6) **{\dir{-}}, (12, 5) *{Z_7},
(4, -11) ; (11, -4) **{\dir{-}}, (12, -4) *{Z_6},
(-10, 0); (10, 0) **@{~},
\end{xy}} \;. \eas

\end{eg}

The diagram represents an integral written in terms of a rational function of the $Z_i$ at the vertices, and a product of distributions, whose arguments are determined by the propagators. The distributions in question are defined as follows.
\begin{dfn}
Write $z^r$ to be the $r^{th}$ component of $\vec{z}$. Then write a modified Dirac delta function \bas \delta^4_p(Z)= (z^p)^4 \prod_{I=1}^4 \delta(Z^I)\;,\eas where $I$ indexes the first four (momentum twistor) coordinates of $Z$.
\end{dfn}
For a $n$ point $N^kMHV$ diagram with $k$ propagators, index the propagators by $p \in \{ 1, \ldots k\}$.
For the $p^{th}$ propagator corresponding to the pair $(i_p, j_p)$ we assign the a function, called the diagram integral
\ba
I(W(\cZ_*) = \int_{R^k} \prod_p \frac{\rd c_{p,0}}{\mathrm{Vol}(Gl(1)) c_{p_0}}\int_{\R^{4k}} \frac{\rd \hat c_{p,i_p}\rd \hat c_{p,i_p+1}\rd\hat c_{p,j_p}\rd\hat c_{p,j_p+1}}{\hat  c_{p,i_p}\hat c_{p,i_p+1}\hat c_{p,j_p}\hat c_{p,j_p+1}}  \delta_p^4(Y_p)\, .
\label{amplitudeeq} \ea
where
\bas
Y_p = c_{p, 0} Z_* + c_{p,i_p} Z_{i_p} + c_{p,i_p+1} Z_{i_p+ 1} + c_{p,j_p} Z_{j_p} + c_{p,j_p+1} Z_{j_p + 1} \in R^{4+k}\, .
\eas

Notice that the vector $Y_p$ is a function of the twistors defining the propagator $p$.

There are a number of details of these formulae that require further explanation
\begin{enumerate}
\item
Here $Z_* \in \R^{4+k}$ is the arbitrary reference twistor mentioned previously, that corresponds to the choice of gauge that leads to MHV diagrams \cite{Adamo:2011pv}. We can choose $Z_*$ such that $\vec{z}_* = 0$. Write $\cZ_* = \left( \begin{array}{c}Z_* \\ \hline \cZ \end{array} \right) \in M(n+1, 4+k)$ to be the augmented matrix defined from $\cZ$ by adding an initial row $Z_*$. Furthermore, we choose $\Z_*$ such that $\cZ_*$ has no maximal minors of determinant $0$. As before, any maximal minor of $\cZ_*$ that only involves the $Z_i$ (i.e. that do not involve the first row) has positive determinant.
\item We write $I(W)(\cZ_*)$ as a shorthand for the fact that the Wilson diagram integralis a function of the twistors forming the rows of $\cZ_*$, namely  $\{Z_*, Z_1, \ldots, Z_n)$.
\item
The Vol$(Gl(1))$ in the denominator is meant in the Fadeev-Popov sense that the integrand has a $Gl(1)$ symmetry under rescaling all the $c_{p,\cdot}$s.  The integral is therefore formally infinite, but can be defined by contracting the 5-form integrand with the generator of $Gl(1)$ and integrating the corresponding 4-form over some 4-cycle that intersects each $Gl(1)$ orbit once.  In practice this can be done by setting one of the $c_{p,\cdot}$s to 1, and the symmetry guarantees that the answer will be independent of the choice.
\item The coefficients $\hat c_{p,s}$ are inductively defined functions of the coefficients $c_{p,s}$ defining $Y_p$. Namely, suppose $\{p_{1} \ldots p_{m}\}$ are the propagators with endpoints on the edge defined by $i$ and $i+1$, ordered inversely to the cyclic ordering of the other endpoints. That is, write $p_{r} = (i, j_r)$, where $j_r >_i j_{r+1}$ in the total ordering on $[n]$ where $i \leq i_1 \ldots \leq i-1$. Then define \bas c_{p_{r}, i} = \hat c_{p_{r-1}, i} \hat c_{p_{r}, i} \quad ; \quad c_{p_{r}, i+1} = \hat c_{p_{r-1}, i+1} \hat c_{p_{r}, i} + \hat c_{p_r, i_1} \;,\eas where $c_{p_{1}, i} = \hat c_{p_{1}, i}$ and $c_{p_{1}, i+1} = \hat c_{p_{1}, i+1}$.
\item
 In order to obtain a rational function, the integrals are performed essentially algebraically against the delta functions allowing us to solve for the $c_{p,\cdot}$s yielding a rational function\footnote{In fact in this real formulation, certain modulus signs will also arise, but we will ignore these in this context; they do not arise in the complex formulation.}. We give more details of this in Section \ref{Wilsonamplitudes}.
\end{enumerate}

Writing the arguments of $Y_p$ in terms of the $\hat c_{p,s}$ enforces a sort of planarity to the diagrams. In particular, if one assumes that the coefficients $c_{p_r, i}$ and $c_{p_r, i+1}$ all have the same sign, for all propagators $p_r$ defined by $i$, this bounds $c_{p_r, i+1}$ on one end by $c_{p_{r-1}, i+1}$. Pictorally, where the endpoints of $p = (i,j)$ are thought of as lying on the line defined by $Z_i, Z_{i+1}$ and $Z_j, Z_{j+1}$, this bound is represented by inserting the end point of $p_r$ further along the line defined by $Z_i, Z_{i+1}$ than $p_{r-1}$.

\begin{dfn}
Write $V_p = \{i_p, i_p+1, j_p, j_p+1\}$ to indicate the set of twistors defining the propagator $p$, and $V^*_p = * \cup V_p$ to be the set including the reference twistor. For a set of propagators $P \in \fP$, write $V_P = \cup_{p \in P}V_p$ and $V_P^* = \cup_{p \in P}V^*_p$ to indicate the set of twistors defining the set of propagators, excluding and including the reference twistor, respectively.
\label{dependencysets}\end{dfn}

In this notation, $Y_p = C_p \cdot \cZ_*$, where $C_p \in \R^{n+1}$ is the vector with entries \bas C_{p, s} = \begin{cases}c_{p, s} & \textrm{if} s\in V_p \\ 0 & \textrm{else}\end{cases}\;. \eas That is, the entries of $C_p$ are the coefficient of $Z_s$ in the representation of $Y_p$ given by the Wilson loop diagram. Each factor of $\delta^4(Y_p)$  implies that the values $c_{p,s}$ define vectors, $C_p$ with are in the kernel of $\cZ_*$. The product of these delta functions implies that we are interested in the span of these vectors.

\begin{dfn}
Let $\fP$ be the set of propagators for a Wilson loop $W$. Write $\cC(W(\cZ)*))$ to be the matrix defined by the row vectors $\{C_p\}_{p \in \fP}$.
\end{dfn}

This matrix defines the span of the vectors $C_p \in \R^{n+1}$. By abuse of notation, we refer to $\cC(W(\cZ)*))$ as the Grassmannian defined by the Wilson loop $W$ and the twistors $\cZ_*$.

In this paper, we are primarily interested in which Wilson loops to study. We mostly concern ourselves with the space $\cC(W(\cZ_*))$ for any Wilson loop. In future work, where we are concerned with the properties of the integrals, $I(W)$, in particular their poles, and in specific physically meaningful sums of diagrams, we explore the integrand and the coefficients $\hat c_{p,r}$, in greater detail.

\subsection{Amplitudes and Grassmannians of Wilson loop diagrams\label{Wilsonamplitudes}}
The diagram integral of Wilson loop diagram defined above by the Feynman diagrams is a map from the twistor configuration space to distributions on the space of supertwistors.

\begin{dfn} For any natural numbers $m, n \in \N$, let $M_{*, +}(m+1, n) \subset M(m+1, n)$ be the subset of $m+1 \times n$ matrices with the property that \begin{enumerate} \item No maximal minors have $0$ determinant, \item All maximal minors that do not involve the first row have strictly positive determinant. \end{enumerate} We call  $M_{*,+}(n+1, k+4)$, with $n  \geq k+4$, the space of twistor configurations. \label{twistconfig} \end{dfn}

Each $\cZ_* \in M_{*,+}(n+1, k+4)$ as defined in section \ref{Feynmanrules} is a twistor configuration.

As shown in equation \eqref{amplitudeeq}, for $W$, a $N^kMHV$ diagram on $n$, the associated integral $I(W)$ is a distribution valued functional mapping from the space of twistor configurations. In this paper, we concern ourselves only with the Grassmannians represented by $\cC(W(\cZ_*))$ that define this map.

However, we wish to avoid the complication of distribution valued amplitudes in this paper. Therefore, we restrict to a well behaved subspace of the twistor configuration space. For $Z = (Z^\mu, \vec{z}) \in \R^{4+k}$, a bosonized supertwistor, let $\pi_4 (Z) \in \R^\mu$ be the four vector defined by the first four (momentum twistor) components of $Z$, as discussed in Section \ref{Feynmanrules}. We write \bas \pi_4 (Z) = Z^\mu \;.\eas

\begin{dfn} \label{genericdfn}
Consider the twistor configuration $\cZ_*$. Define a matrix $\cZ_*^\mu$, where each row is given by the projection  $\{Z_*^\mu, Z_1^\mu , \ldots, Z_n^\mu\}$. The twistor configuration $\cZ_*$ is called generic if $\cZ_*^\mu \in M_{*, +}(n+1, 4)$.
\end{dfn}

Given a fixed twistor configuration, $\cZ_*$, the propagators in $W$ define a subspace of the kernel of $(\cZ_*^\mu)^T$, explicitly, that spanned by $\{C_p\}_{p \in \fP}$. The matrix $\cC(W(\cZ_*)) \subset \ker \cZ_*^\mu$ represents said subspace. In other words, we may view a diagram, $W$, as a map from twistor configurations to Grassmannians, \ba W : M_{*,+}(n+1, 4+k) & \rightarrow \bigoplus_{d =1}^{n-3} G(d, n+1) \label{Wilsonmap}\\ \cZ_* & \rightarrow \cC(W(\cZ_*)) \nonumber \;.\ea The integral, $I(W) (\cZ_*)$, assigns a function of $\{Z_*, Z_1, \ldots Z_n\}$ to the Grassmannian $\cC(W(\cZ_*))$. We may evaluate the form of this matrix explicitly.

\begin{dfn}Denote by $\cZ_*^\mu|V_p$, the minor of $\cZ_*^\mu$ defined by the set $V_p$. Define the determinant \bas \detzr{} =  \det (\cZ_*^\mu|V_p)\;. \eas Write $\detzr{m}$ to be the determinant of the matrix formed by replacing the vector $Z_m^\mu$ of $\cZ_*^\mu|V_p$ with $Z^\mu_*$. The quantity $\detzr{m}$ is not defined if $Z_m^\mu$ is not a row of $\cZ_*^\mu|V_p $.  \end{dfn}

The expression $\delta^4_p(Y_p)$ is non-zero when $C_p$ defines the kernel of the matrix \ba (\cZ_{*}^\mu|V_p^*)^T\in M(4, 5), \label{propagatorkernel}\ea By Cramer's rule, under the above definition, integrating against the functions $\delta^4_p(Y_p)$ sets \ba \frac{c_{p, m}}{c_{p,0}} = - \frac{\detzr{m}}{\detzr{}} \label{cfunctions}\;,\ea for $m \in V_p$. Since the kernel defined by $\delta^4_p(Y_p)$ is one dimensional, we are free to set $c_{p,0}$ as a free variable.

The dimension of $\cC(W(\cZ_*))$ as a Grassmanin, i.e. the rank of the matrix, is bounded above both by $k = |\fP|$, the number of propagators of $W$, and by $|V(\fP)|-3$. Recall from Definition \ref{dependencysets} that $|V(\fP)| \leq n+1$ is the number of twistors involved in defining the propagators of $W$. Then $|V(\fP)|-3$ is a bound for the rank of $\cC(W(\cZ_*))$, as it is the dimension of $\ker (\cZ_*^\mu|V(\fP))T$.

We are now ready to explicitly calculate the diagran Integral defined by a Wilson loop diagram. For the sake of simplicity, we consider a diagram that does not have two propagators ending on the same edge. In this case, $\hat c_{p,s} = c_{p,s}$ for all $p$ and $s$. The integrand for a Wilson loop diagram with multiple propagators ending on the same edge is derived similarly, with the added complication of solving for $c_{p,s}$ in terms of the$\hat c_{p,s}$.

For a such a Wilson loop diagram, integrating against the variable $c_{p,i}$, with $i\neq 0$ sets evaluates the $c_{p,s}$ according to \eqref{cfunctions}. It remains to evaluate the integrals of the form \bas  \int_{\R}  \frac{\rd c_{p,0}}{\mathrm{Vol}(Gl(1)) c_{p,0}}\;.\eas The factor of $\mathrm{Vol}(Gl(1))$ in the denominator of this integrand allows us to set $c_{p, 0}$ to a constant of our choice. For ease of notation, we write $c_{p, 0} = - \detzr{}$. Then equation \eqref{cfunctions} gives $c_{p,m} = \detzr{m}$.

To calculate the numerator, recall that there is also a twistor component to $\delta_p(Y_p) = (y^p)^4 \prod_{I=1}^4 Y_p^I$, where \bas y^p = c_{p, i_p}z_{i_p}^p + c_{p, i_p+1}z_{i_p+1}^p + c_{p, j_p}z_{j_p}^p + c_{p, j_p+1}z_{j_p+1}^p \eas is the bosonized component of the twistor $Y_p$. Recalling that $Z_*$ is chosen such that $\vec{z}_*=0$, the integral given in equation \eqref{amplitudeeq}, evaluates to \bas I(W(\cZ_*)) =  \frac{(\detzr{i_p}z_{i_p} + \detzr{i_p+1}z_{i_p+1} +\detzr{j_p}z_{j_p}+ \detzr{j_p+1}z_{j_p+1})^4}{\detzr{}\detzr{i_p}\detzr{i_p+1}\detzr{j_p}\detzr{j_p+1}} \;\eas for $W$ a Wilson loop diagram with no two propagators sharing a boundary edge.

For physical reasons, we are interested only in the Grassmannians with certain properties on $\cC(W(\cZ_*))$.

\begin{dfn}Define $\cM(W(\cZ_*)) \in M_{k,n}$ to be the matrix derived from $\cC(W(\cZ_*))$ by ignoring the first column. \label{mdef} \end{dfn}

We are interested in Wilson loops $W$ such that $\cM(W(\cZ_*))$ is a non-negative Grassmannian:  \bas \cM(W(\cZ_*)) \in G_\geq (k, n) \;.\eas In particular, $\cM(W(\cZ_*))$ must have rank $k$.

In light of this, we make the following definition.

\begin{dfn}
A Wilson loop diagram, $W$, is admissible if there is a generic twistor configuration such that $\cM(W(\cZ_*))\in G_{\geq 0}(k, n)$. \label{admissibledfn}
\end{dfn}

That is, the matrix $\cM(W(\cZ_*))$ has full rank, and all maximal minors have non-negatives determinants for some generic twistor configuraton.

In this paper, we study the Grassmannians $\cM(W(\cZ_*))$ defined by a Wilson loop diagram, as defined by $W(\cZ_*)$ using the language of matroids. We are only interested in Wilson loop diagrams In the next section, we further classify Wilson loop diagrams, to eliminate a large class of inadmissibly diagrams from consideration entirely.

\subsection{Admissible Wilson loop diagrams\label{welldefinedsect}}

In the previous section, we show that Wilson loop diagrams define a subspace of the kernel of the matrix $(\cZ_*^\mu)^T$, for a given twistor configuration e $\cZ_*$. In this section, we examine the properties of that subspace.

\begin{dfn}\label{overexact}
Any diagram $W= (\fP, [n])$ that has a set of propagators $P \subset \fP$, such that \ba |V_P| < |P|+3 \label{overdefinedcond}\ea is called overdefined. If $W$ is not overdefined, it is well defined. If a well defined diagram contains a set of propagators $P$ such that \ba |V_P| = |P|+3 \label{exactcond}\ea it is called exact.
\end{dfn}

Notice that any subdiagram of $W$, $(P, V_P)$, satisfying \eqref{overdefinedcond} is, for physical reasons, known to be trivial. Any subdiagram $(P, V_P)$ satisfying \eqref{exactcond} is an $\overline{MHV}$ diagram, and uninteresting, as one knows that the sums of all such diagrams, for a fixed $|P|$ is trivial. However, there is no generalization of this to all diagrams, even those at $n > k+4$, that contain such interactions as subdiagrams. In this section, we show that overdefined Wilson loops are not admissible, and that exact diagrams, in some sense, uniquely define their Grassmannians.

We begin with overdefined Wilson loops.

\begin{thm}\label{overdefinethm}
If $W$ is an overdefined Wilson loop, then it is not admissible.
\end{thm}

\begin{proof}
It is enough to show that if $W$ is overdefined, then $\cM(W(\cZ_*)$ does not have full rank for a generic twistor configuration.

For this calculation, we work in a basis $\{e_*, e_1 \ldots e_n\}$ of $\R^{n+1}$.

Let $P$ be a set of propagators of $W$ satisfying \eqref{overdefinedcond}. Suppose $|V_P|= m$. Let \bas U_P = \spn\langle\{e_j\}_{j \in V_P^*}\rangle\eas be the $m+1$ dimensional subspace of $\R^{n+1}$ coressponding to the twistors defining the propagators in $W$. Write $\cZ_*^\mu|V_P^* \in M_{*,+}(m+1, 4)$ to be the momentum twistor matrix of the twistors defining the propagators in $P$. Write $C_p|U_P$ to be the projection of $C_p \in \ker \cZ_*^\mu$ onto $U_P$. Restricted to this vector space, $\dim (\ker \cZ_*^\mu)|U_P = m-3$. If $|P| > m-3$, then the set $\{C_p|U_P\}$ are not all independent.

Since the propagators $p\in P$ do not depend on any vertex outside of $V_P^*$, $C_p|U_P^\perp = \vec{0}$. This implies that the $C_p$ are not all independent. The matrix $\cM(W(\cZ_*))|V_P^*$ does not have full rank in any generic twistor configuration. Therefore, $W$ is not admissible.
\end{proof}

In fact, we have shown that for an overdefined Wilson loop, the matrix $\cM(W(\cZ_*))|V_P^*$ does not have full rank in a generic twistor configuration. Below we show that the converse is also true. Given a generic twistor configuration, the matrix $\cM(W(\cZ_*))$ has full rank only if the Wilson loop diagram is well defined.

The generic twistor condition is essential to the arguments in this paper. Without this requirement, it is possible to choose a twistor configuration $\cZ_*$ such that the associated matrices, $\cM(W(\cZ_*))$, have maximal minors with $0$ determinants. However, for well defined Wilson loops, such a twistor configuration corresponds to a low dimensional hypersurface in twistor space. These hypersurfaces lead to poles of the holomorphic versions of the integrals $A(W(\cZ_*))$ expressed in \eqref{amplitudeeq}. Identifying these hypersurfaces for families of Wilson loops and studying the structures of the poles at these hypersurfaces is the the subject of ongoing study in the physics \cite{Lam14} and also the subject of future work for these authors.

\begin{thm}
Given a Wilson loop $W = (\fP, [n])$, and a $\cZ_*$, a generic twistor configuraton,  the matrix $\cM(W(\cZ_*)) \in G(|\fP|, [n])$, if and only if $W$ is well defined.  \label{genericindep}
\end{thm}

\begin{proof}
It is equivalent to show that the matrix $\cM(W(\cZ_*))$ has full rank if and only if $W$ is well defined.

In Theorem \ref{overdefinethm}, we have shown that, for a generic twistor condition, if $W$ is not well defined, then it does not have full rank.

Suppose $W$ is a well defined Wilson loop with $k$ propagators, but that $\cM(W(\cZ_*))$ does not have full rank. This implies that there exists $p \in \fP$ such that \bas C_p \in \spn \langle \{ C_q \}_{q \in \fP \setminus p}\rangle\; .\eas Since $C_p = \ker \cZ_*^\mu|V_p^*$, this implies that \bas \spn\langle \{\cZ^\mu_i\}_{i \in V_p^*} \rangle \subset \spn\langle \{\cZ^\mu_i\}_{i \in V_{\fP \setminus p}^*} \rangle\;. \eas

In other words, the twistor configuration $\cZ_*$ is not generic.
\end{proof}

Henceforth, we only consider well defined Wilson loops. These have the further property that they are closed under taking subdiagrams.

\begin{dfn}
Consider a Wilson loop $W = (\fP , [n])$. For any subset of propagators, $P \subset \fP$, and set of vertices, $S$, such that $V_P \subset S \subset [n]$, the Wilson loop diagram $(P, S)$ is a subdiagram of $W$, where a cyclic ordering on $S$ is induced from the cyclic ordering on $[n]$.
\label{subdiagramdfn}\end{dfn}

For ease of future reference, we define a family of subdiagrams we refer to frequently in the sequel. For $W = (\fP, [n])$, and $P \subset \fP$, \bas W|P := (P, [n])\;.\eas

We show that Theorem \ref{genericindep} extends to all subdiagrams of $W$ of the form $W|P$.

\begin{cor}
Consider a well defined Wilson loop, $W$, and a generic twistor configuration $\cZ_*$. The matrix $\cM(W|P(\cZ_*)) \in G(|P|, n)$. I.e., it has full rank.
\label{subdiagramcor}\end{cor}

\begin{proof}
If $W$ is a well defined Wilson loop, then there does not exists any subset $P$ such that $|V_P| \leq |P| + 2$, by defintion \ref{overexact}. Therefore, any subdiagram of $W$ is also well defined.
\end{proof}

Finally, we consider exact Wilson loops. We define an equivalence relation on exact Wilson loops.

\begin{dfn}\label{equivdiagramdef}
Two exact Wilson loops, $W = (\fP, n)$ and $W'= (\fP', n)$ are equivalent if
\begin{enumerate}
\item There exist sets of propagators $P_i$ in $W$ and $P_i'$ in $W'$ such that $|P_i| + 3 = |P_i'| + 3 = |V_{P_i}|$.
\item The sets defining these propagators $V_{P_i} =  V_{P_i'}$ are equal.
\item The remaining propagators of both diagrams are the same, $\fP \setminus \cup_i P_i = \fP' \setminus \cup_i P_i'$.
\end{enumerate}
\end{dfn}

It is clear that this is an equivalence relation. This equivalence relation leads to a question of some importance in understanding the nature of the integrals associated to Wilson loop diagrams. This is discussed in further detail after Theorem \ref{noncrosspos}.

\begin{thm}
If two exact Wilson loop diagrams are equivalent, and $\cZ_*$ is generic, the diagrams define the same subspace of $\ker \cZ_*^\mu$
\label{equivalencethm}\end{thm}

\begin{proof}
Suppose there are two sets of propagators, $P$ in $W$ and $P'$ in $W'$ such that \bas |V_P|= |P| + 3 = |P_i'| + 3 = |V_{P'}|\;. \eas If there are more such sets, apply the arguments below separately to each.

Suppose $|V_P|= m$. Write $\cZ_*^\mu|P \in M_{*,+}(m+1, 4)$ to be the momentum twistor matrix of the twistors defining the propagators in $P$. As in Theorem \ref{overdefinethm}, define $U_P = \textrm{span}\langle\{e_j\}_{j \in V_P^*}\rangle$. Then $\{C_p|U_P\}_{p \in P} \in \R^{m+1}$ are the vectors in $\ker \cZ_*^\mu$ defined by the propagators, $P$. Note that $\dim \ker \cZ_*^\mu|P = m-3$. Since $|P| = m-3$ and $W$ is not overdefined, the vectors $C_p|U_P \in \R^{m+1}$ are linearly independent and span $\ker \cZ_*^\mu|P$, as in Theorem \ref{genericindep}.

Similarly, in $W'$, the vectors $C_{p'}|U_{P'} \in \R^{m+1}$ are linearly independent and span $\ker \cZ_*^\mu|P$. Since $V_P = V_{P'}$ the two sets of vectors define the same vector space in $\R^{n+1}$, \bas \spn \langle \{C_{p'}\}_{p' \in P'} \rangle = \spn \langle \{C_{p}\}_{p \in P}\rangle  \;. \eas

This holds for each pair of sets of propagators, $P_i \subset \fP$ and $P_i' \subset \fP'$ that satisfy \eqref{exactcond}.

Since the remaining propagators in $W$ and $W'$ are the same, the subspaces of $\ker \cZ_*^\mu$ defined are the same.
\end{proof}

In the remainder of this paper, we leave the matrix representations of Wilson loops behind, and study only the combinatorics of the diagrams, using the language of matroids. In doing so, we obtain the following results.

Theorem \ref{noncrosspos} shows that any well defined Wilson loop diagram with non-crossing propagators is admissible.

This last point goes against the positivity conjecture for Wilson loops, which states that a Wilson loop is planar if and only if the associated Grassmanians are positive. We give a partial solution to this problem. Theorem \ref{equivalencethm} shows that any admissible exact Wilson loop diagram with crossing propagators is equivalent to a Wilson loop diagram with non-crossing propagators. Therefore, at least in the exact case, an exact Wilson loop diagram with crossing propagators is admissible if and only if it is equivalent to an exact diagram with non-crossing propagators.

\section{The Matroidal language \label{matroiddefssection}}

In this section, we review the concept of matroids. The results and definitions set forth in this section are not new. For a more comprehensive review on of the material, see \cite{flacets05, positroids13, WelshMatroid}. In the most abstract sense, a matroid is a set of independency data on a set. It is a generalization of the concept of a matrix. However, a \emph{realizable} matroid, which is the only sort we examine in this paper, is a set of independency data that can be represented by a collection of vectors in $\R^n$, i.e., by a Grassmannian. The natural objects to study non-negative Grassmannians are positroids, a subclass of realizable matroids that can be realized by non-negative Grassmannians.

In this paper, we use caligraphic script ($\cM$) to denote matrices and plain text ($M$) to denote matroids and Grassmannians.

As a matroid only considers the independency data, a realizable matroid can be realized by a family of Grassmannians of a form that makes it ideal for studying Wilson loop diagrams. In particular, given two different generic twistor configurations $\cZ_*$ and $\cZ_*'$, the matrices $\cM(W(\cZ_*))$ and $\cM(W(\cZ_*'))$ define different Grassmanians. However, as suggested by Theorems \ref{overdefinethm} and \ref{genericindep}, the independence of the colums is determined by the propagator structure. Therefore, $\cM(W(\cZ_*))$ and $\cM(W(\cZ_*'))$ define the same matroid.

In this section, we draw parallels between the matroidal concepts and Grassmannians whenever possible. The advantage of matroids over Grassmannians is that in their component data, sets and collections of subsets, lends itself easily to combinatorial approaches. Therefore, when studying combinatorial or diagramatic objects that represent Grassmannians, such as Wilson loops, one may apply the combinatorics inherent in matroids to the diagramatics of the desired system without ever having to study the actual associated Grassmannians.

\begin{dfn}
A matroid, $M$, is given by a pair of sets $(E, \fB)$, where $E$ is a finite set, called the ground set of $M$. The set $\fB$ is a set of subsets of $E$ with the property that if $B_1, B_2 \in \fB$, there exists $b_1 \in B_1 \setminus B_2$ and $b_2 \in B_2 \setminus B_1$, and $(B_1 \setminus b_1) \cup b_2 \in \fB$.
\end{dfn}

Notice that all elements of $\fB$ are the same size. This is the rank of $M = (E, \fB)$, denoted $\rk{M}$. For a general subset $S \subset E$, and basis $\fB$ defining a matriod, we say that the rank of $S$ is $\rk (S) = \max\{|B \cap S| | B \in \fB\}$. Furthermore, any subset of the ground set, $S \subset E$, is independent if and only if there is a basis set containing it, $S \subset B \in \fB$.

\begin{dfn}Let the matroid $M$ have rank $d$; $\rk(M) = d$. If the ground set $E$ can be realized by a set of vectors $\{a_1 \ldots a_n \}\in k^d$, for some field $k$, with each subset $\{a_i\}_{i \in B}$ for $B \in \fB$ is linearly independent, then the matroid $(E, \fB)$ is realizable over $k$.\label{realizabledef}\end{dfn}

In other words, a realizable matroid can be represented by a set of vectors $a_1, \ldots, a_n \in k^d$, or as matrix $\cA$ over $k$, $\cA \in M_k(d,n)$. However, this is not the unique realization of the matroid. Any set of $n$ vectors in $k^d$ satisfying the independency data laid out by the set $\fB$ realizes $(E, \fB)$. In particular, the set of realizing matrices is invariant unter a $GL_k(d)$ action. Therefore, it is natural to think of a matroid to be realizable by a family of Grassmanians in $G(d, n)$.  All the matroids considered in this paper are realizable.

\begin{eg}\label{coimagegrassmannian}
Consider a realizable matroid $M = (E, \fB)$, with $|E| = n$ and $\rk (E) = d$. Let $\{e_i\}_{i \in E}$ be the set of basis vectors for $k^n$. Then $M$ can be realized as a matrix, $\cM$, mapping from $k^n$ to $k^d$. Write $V_{\cM} = \textrm{ coimage }\cM \subset k^n$. In this manner, a matroid defines a family of $d$ planes in $k^n$, $\{V_{\cM} | \cM \textrm{ realization of } M\} \subset G(d,n)$, corresponding to the Grassmannians that realize it. A basis, $B \in \fB$, is realized by a non-zero minor of $\cM$. The collection $\fB$ is the set of all non-zero minors of $\cM$.
\end{eg}

For the purposes of this paper, we are interested in non-negative Grassmannians. This too can be captured by matroids.

\begin{dfn}
A positroid is a realizable matroid with the following data $([n], \fB)$, where $[n]$ is a cyclically ordered set, that has a realization in the non-negative Grassmanians.
\end{dfn}

In other words, a positroid is a matriod for which there exists a non-negative Grassmannian with the relevant independency data.

There are a few natural operations on matriods. First, we define the dual of a matroid.

\begin{dfn}
For $M = (E, \fB)$, the dual matroid is defined $M^* = (E, \fB^*)$, where $\fB^* = \{E \setminus B |B \in \fB\}$.
\end{dfn}

\begin{eg} Given a matroid $M$, and a matrix, $\cM$ that realizes it, as in Example \ref{coimagegrassmannian}, the dual matroid can be realized by the Grassmanian of the form $V_{\cM}^\perp$. This is, of course, the kernel, $\ker \cM$. Concretely, the dual matroid $M^*$ is defined by the independency data of the matrix $(\ker \cM)^T$.
\label{dualkernel} \end{eg}

There are two dual operations on $M$: restriction and contraction.

\begin{dfn}
Consider the subset $S \subset E$. The restriction of the matroid $M$ to $S$, $M|S$, is defined by the sets $(S, \fB|S)$, where \bas \fB|S = \{B\cap S | \;|B\cap S|\textrm{  maximal among }B\in \fB \}\;. \eas
\end{dfn}

The dual operation, contraction, is defined as follows:

\begin{dfn}
Consider the subset $S \subset E$. The contraction of the matroid $M$ by $S$, $M/S$, is defined by the sets $(E \setminus S, \fB/S)$, where \bas \fB/S = \{B\setminus S | \;|B\cap S| \textrm{ maximal among }B\in \fB\} \;. \eas
\end{dfn}

These two operations are dual in the following sense.

\begin{prop} \label{restrictconstractdual}
Given a matriod $M = (E, \fB)$, and a set $S \subset E$, \bas (M/S)^* = M^*|(E \setminus S) \;. \eas
\end{prop}

\begin{eg} \label{restrictcontracteg}
Consider a realizable matroid $M= (E, \fB)$ with $|E| = n$. Consider $k^n$, with $E$ the basis set. For any subset $S \subset E$, define the vector space $V_S\subset k^n$ spanned by the basis vectors $\{e_i\}_{i\in S}$. Then, for $\cM$, a realization of $M$, the matrix $\cM|S$, defined by the columns of $\cM$ corresponding to $S$, is a realization of the matroid $M|S$. The coimage of this matrix is given by the projection of $V_{\cM}$ onto $V_S$: $V_{\cM|S} := proj_{V_S} (V_\cM) = proj_{V_S}\textrm{ coimage } \cM|S \subset V_S$.

By Proposition \ref{restrictconstractdual}, the matroid $M/S$ is realized by the subspace of $V_S^\perp$, $(proj_{V_S^\perp} V_{\cM}^\perp)^\perp) \subset V_S^\perp$.
\end{eg}

Matroid data is essentially about the set of independent subsets of a given set. To this end, there are several important concepts to keep in mind.

\begin{dfn} \label{matroiddefs}
Consider a matroid $M = (E, \fB)$.
\begin{enumerate}
\item  The set $F \subset E$ is a flat if it is a maximally dependent subset of $E$. Equivalently, $\forall e \in E \setminus F$, $\rk (F \cup e) = \rk (F) +1$.
\item The set $C \subset E$ is a circuit if it is a minimially dependent set. That is, $C$ is not contained in any $B \in \fB$, but, for any $e \subset C$, there is a $B \in \fB$ such that $C \setminus e \subset \fB$ .
\item A cyclic flat is a flat that can be written as a union of circuits: $F = \cup_i C_i$, for $C_i$ circuits.
\item A matroid is connected if and only if there is not a partition of $E$, $E = \amalg_i E_i$, $E_i \cap E_j = \emptyset$ for $i \neq j$ such that all the elements of each $E_i$ are mutually independent. That is, given any collection $\{S_i \in E_i \}$ of independent sets, the unition $\cup_i S_i$ is also independent.
\end{enumerate}
\end{dfn}

In the following example, we illustrate a few of these definitions.

\begin{eg}
Let $M= (E, \fB)$ be a realizable matroid as before, where $E$ is the set of basis vectors for $k^n$. A set $F$ is a flat if and only if, for every element $e \in E\setminus$, $\rk M|F <  \rk M|(F\cup e)$.

A circuit is simply a set $C$ such that $\rk C = |C| -1$. The nomenclature comes from the point of view that matroids are generalizations of graphs, which is not addressed in this paper.

The direct sum of two matroids, $M_1= (E_1, \fB_1)$  and $M_2= (E_2, \fB_2)$ is written \ba M_1 \oplus M_2 = (E_1 \amalg E_2, \fB_1 \cup \fB_2)\;. \label{directsumeq}\ea

The intersection of two flats is a flat. Indeed, for two flats $F_1$ and $F_2$ of $(E, \fB)$, any element $e \in (F_1 \cap F_2)^c$ is either not an element of $F_1$ or not an element of $F_2$. Therefore, $\rk ((F_1 \cap F_2) \cup e) =  \rk (F_1 \cap F_2) + 1$.

\end{eg}

Next, we note a result about cyclic sets of disconnected matroids which we use later.

\begin{lem}
Let $C$ be a cyclic set of $M$. If the matroid $M|C$ is disconnected, then $C$ can be partitioned into $C = C_1 \amalg C_2$, where $C_1$ and $C_2$ are cyclic flats.
\label{disconnectedcyclic}\end{lem}

\begin{proof}
If $M|C$ is disconnected, then $M|C = M|C_1 \oplus M|C_2$, where $C_1$ and $C_2$ are mutually independent. In other words, they are flats of $M|C$. Since $C$ is a cyclic set, one can write $C = \cup S_i$, where each $S_i$ is a circuit. We claim that, for any such covering of $C$ by circuits, there cannot be a circuit $S$ that intersects both $C_1$ and $C_2$ non-trivially.

In fact, if there were such an $S$, consider the element $e_1 \in C_1 \cap S$. Since $S \setminus e_1$ and $S \cap C_1$ are both independent, $\rk ((S \setminus e) \cap C_1) = |S \cap C_1| -1$. But $S$ is a circuit. Therefore, $\rk (S) = \rk S\setminus e_1$. This implies either that $S_1 \cap C_1$ is a dependent subset of $C_1$ (if $\rk (S \cap C_1) = |S \cap C_1| -1$), violating the minimal dependence of circuits, or that $e_1$ is not independent of $S \cap C_2$ (if $\rk (S \cap C_1) = |S \cap C_1|$), contradicting the fact that $M|C$ is disconnected.

Therefore, any covering of $C$ by circuits can be separated into coverings of $C_1$ and $C_2$ by circuits, implying that they are both cyclic fats.
\end{proof}

We have seen that a matroid, $(E, \fB)$, defines a subset of Grassmannians over the field $F$, $G(\rk(E), |E|)$. There is a geometric interpretation of matroids in $k^{|E|}$, with basis $\{e_i \}_{i \in E}$ as usual.

\begin{dfn}
Given a connected matroid $M = (E, \fB)$, let $F^E$ be the vector space indexed by the elements of $E$, as usual. For $B\in \fB$, the indicator vector of the basis is defined \bas e_B = \sum_{s \in B} e_s \;.\eas The matroid polytope is defined \bas \Gamma_M = \textrm{convex hull}(e_B | B \in \fB) \;.\eas
\end{dfn}

This polytope has the property that every codimension $1$ face of a matroid polytope is also a matroid polytope. In particular, the codimension one faces, or facets, of the matrix polytope is defined by a subset $S \subset E$, called a flacet.

\begin{thm} \label{flacetcondition}
Let $M = (E, \fB)$ be a connected matroid. The facets of $\Gamma_M$ are defined by flacets. A subset $F \subset E$ is a flacet if an only if $M|F$ and $M/F$ are both connected.
\end{thm}

\begin{proof}
This is proved in \cite{flacets05}. The proof is not reproduced here.
\end{proof}

It is worth noting one fact about flacets. This is well known in the matroid literature, but worth going through as an exercise for those not familiar with matroid calculations.

\begin{prop} \label{flacetcyclicflat}
All flacets are cyclic flats.
\end{prop}

\begin{proof}
Let $M$ be a matroid with ground set $E$. Let $F$ be a flacet of $M$. For $S \subset E$ as subset of the groundset, let $S^c = E\setminus S$ indicate the set complement.

If $F$ is not a flat, consider $G$, the smallest flat containing $F$. We claim that $M/F$ is not connected. Since one is contracting by $F$, no element of $G\setminus F$ lies within a basis of $M/F$. If it did, that would violate the maximality of $B \cap F$. Therefore, we may write \bas M/F = M/F|_{(G\setminus F)c} \bigoplus M/F|_{(G\setminus F)}\;, \eas where $M/F|_{(G\setminus F}$ is a matroid of rank $0$, showing $M/F$ to be disconnected.

The condition that $M|_F$ is connected is equivalent, by Proposition \ref{restrictconstractdual}, to requiring $(M^*/F^c)^*$, and thus $M^*/F^c$ to be connected. This implies that $F^c$ is a flat in $M^*$. We claim that if $F$ is a flat in $M$ if and only if $F^c$ is cyclic in $M^*$. Therefore, if both $M/F$ and $M|_F$ are connected, then $F$ is a cyclic flat.

To see the claim, check that if $F$ is a circuit in $M$, then $F^c$ is a flat in $M^*$. As above, write $M = (E, \fB)$ and $M^* = (E, \fB^*)$. The set $F$ is a circuit if and only if, for any $e \in F$, there is a $B \in \fB$ such that $F \setminus e \in B$. That is, \bas (F\setminus e)^c = F^c \cup e \subset B^c \in \fB^* \;.\eas This shows that for all $e \not \in F^c$, $e$ is linearly independent of $F^c$, making $F^c$ a flat of $M^*$. Let $C = \cup F_i$ be a cyclic set in $M$, for $F_i$ circuits. Then $C^c = \cap_i F_i^c$ is an intersection of flats. Since the intersection of flats is a flat, $C^c$ is a flat. By exchanging the roles of $M$ and $M^*$, one sees that $F$ is a flat in $M$, if and only if $F^c$ is cyclic in $M^*$, as desired.
\end{proof}

In fact, we may say more than this.

\begin{cor}
Let $M$ be a matroid with ground set $E$. If a subset $S \subset E$ cannot be partitioned into cyclic flats then the restricted matroid $M|S$ is connected if and only if $S$ is cyclic.
\label{partialcyclicequiv}\end{cor}

\begin{proof}
Arguments in Theorem \ref{flacetcyclicflat} show that if $M|S$ is connected, then $S$ is cyclic. Lemma \ref{disconnectedcyclic} shows that if $S$ cannot be partitioned into cyclic flats then the restricted matroid $M|S$ is either connected or $S$ is not cyclic. These two combine to give the desired result.
\end{proof}

We use this notation to define the matrix polytope by a set of inequalities and constraints.

Finally, one can use flacets to determined whether a matroid is a positroid or not. First, we recall that a positroid is defined only on a cyclically ordered set $[n]$.

\begin{dfn}
A cyclic interval is a set
\bas [k, l] = \begin{cases} \{k, k+1 \ldots l-1, l \} & k < l \\ \{k, k+1 \ldots n, 1 \ldots l-1, l \} & l < k \end{cases} . \eas
\end{dfn}

This gives the following theorem about which matroids are are positroids.

\begin{thm} \label{positroidcondition}
A connected matrioid $M = (E, \fB)$ is a positroid if and only if every flacet is a cyclic interval.
\end{thm}

\begin{proof}
This is proved, though using slightly different nomenclature, in Proposition 5.6 of \cite{positroids13}.
\end{proof}

It is sufficient to consider only connected matroids in order to identify positroids.

\begin{thm}
For $M = ([n], \fB)$ a disconnected matroid, write $M = M_1 \oplus \ldots M_k$, where $M_i = (E_i, \fB_i)$. The matroid $M$ is a positroid if and only if the sets $\{E_1 \ldots E_k\}$ form a non-crossing partion of $[n]$, and each $M_i$ is a positroid. \label{disconnpositroid}
\end{thm}

To understand the statement of this theorem, we define non-crossing partitions.

\begin{dfn}
A non-crossing partition of $[n]$ is a set of sets $\{E_1 \ldots E_k\}$ that partition $[n]$, that is, $E_i \cap E_j = \emptyset$ for $i \neq j$, and $\amalg_{i = 1}^k E_i = [n]$ with the further condition that, if the set $[n]$ labels points on a circle, and $P_i$ corresponds to the inscribed polygons defined by $E_i$, then none of the polygons intersect, (that is, none of their lines cross).
\label{noncrossingpartitiondfn}\end{dfn}

\begin{proof}
The proof of this theorem is given in section $7$ of \cite{positroids13} and is not repeated here.
\end{proof}

Having defined positroids, we use these definitions to discern, graphically, which Wilson loop diagrams correspond to positive Grassmannians. This is the subject of the section \ref{Wilsonloopstopositroids}.

\section{Wilson loop diagrams and positroids \label{Wilsonloopstopositroids}}

In the previous section, we introduced the idea of positroids, and criteria for determining when a matroid is a positroid. In this section, we return to the definitions of Wilson loop diagrams and their related matroids, as realized by the matrices in Section \ref{Feynmanrules}. We address the question of how to determine if a Wilon loop diagram is admissible. This requires two conditions. One that it is well defined, the second is, that it corresponds to a positroid. In fact, we show that a well defined Wilson loop diagram with non crossing propagators is a positriod. We also show that exact Wilson loop diagrams with crossing propagators only in the purely exact subdiagrams can be "untangled". Therefore, these too define positroids. This is done in several steps. The first section, \ref{Wilsonspace}, defines a map from Wilson loops to matroids realized over $\R$.

\subsection{The space of Wilson Loop Diagrams\label{Wilsonspace}}

In this section, we define the set of Wilson loops. Section \ref{Wilsondefsect} defines the set of well defined Wilson loop diagrams consistent with those developed in Section \ref{physics} and translates the physically motivated diagrams into combinatorially defined objects. Section \ref{thematroid} defines the matroid $M(W)$, defined by the matrix $\cM(W(\cZ_*))$, from a Wilson loop diagram.

\subsubsection{Combinatorial Wilson Loop diagrams\label{Wilsondefsect}}
A well defined Wilson loop diagram is comprised of an $n$ sided boundary polygon, and $k$ propagators with endpoints on the edges of the polygon. These $k$ edges are all drawn in the interior of the polygon. As the edges of the polygon and the propagators are intrinsically different edges, we denote the propagators by wavy lines.

\begin{dfn}
A well defined Wilson loop diagram consists of the following data.
\begin{enumerate} \item Label the vertices of the polygon by the cyclic set $[n]$, by making a choice of which vertex to label $1$. The ordering of the vertices is always counterclockwise. In the diagrams in the sequel, the first vertex is labeled by a bullet, and the other labels inferred.
\item Label the edges of the polygon $\{e_1 \ldots e_n\}$, where $e_i$ corresponds to the edge connecting the vertex $i$ and $i+i$. It is also useful to think of an edge $e_i$ as the pair $(i, i+1)$.
\item Let $\fP$ be the set of propagators of a Wilson loop diagram. A propagator $p \in \fP$, is the set of ordered pairs, $\{(i_p, j_p)\}$ indicating the edges defining the endpoints of $p$.
\item Label the propagators $p_1 \ldots p_k$ in a counterclockwise order in the diagram, that is, let $i_{p_j,1} \leq i_{p_{j+1},1}$.
\item Let $V_p = i_p, i_p+1, j_p, j_p+1$ be the dependency set of the propagator $p$. Then, for a set of propagators $P$, the dependency set is $V_P = \cup_{p \in P}V_p$.
\item There is no set of propagators $P \in \fP$ such that $|P| +2 \geq |V_P|$.
\end{enumerate} \label{Wilsondef}
\end{dfn}
Given this notation, we write a Wilson loop diagram $W=(\fP, n)$. The last condition restricts all Wilson loop diagrams studied in this section to be well defined, under definition \ref{overexact}.

\begin{eg} \label{disconnected3eg}
Below is an example of a Wilson loop diagram. Consider the Wilson loop diagram $W = (\{(2, 4),(4, 7), (5,7)\},[8])$.
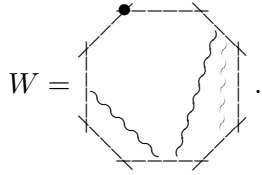
\bas W =
{\begin{xy}
(-5,10) *{\bullet},
(-6, 10); (6, 10) **{\dir{-}},
(-4, 11); (-11, 4) **{\dir{-}},
(-10, 6); (-10, -6) **{\dir{-}},
(-11, -4); (-4, -11) **{\dir{-}},
(-6, -10); (6, -10) **{\dir{-}},
(4, 11); (11, 4) **{\dir{-}},
(10, 6); (10, -6) **{\dir{-}},
(11, -4); (4, -11) **{\dir{-}},
(-10, 0); (0, -10) **@{~},
(2, -10); (7, 8) **@{~},
(8, -7); (8, 7) **@{~},
\end{xy}} \;. \eas

\end{eg}

For the rest of this paper, there are two important classes of Wilson loop diagrams to consider: those with crossing propagators, and those with non-crossing propagators.

\begin{dfn} \label{crossingpropagatordef}
Let $p$ and $p'$ be two propagators of $W(\fP, n)$. We say that $p$ and $p'$ are crossing propagators if $i_{p,1} < i_{{p'},1}$ and $i_{p,2} < i_{{p'},2}$. A Wilson loop is a planar diagram if it has no crossing propagators. Otherwise, it is a non-planar diagram.
\end{dfn}

This nomenclature comes from the fact that any graphical depiction of a Wilson loop with crossing propagators must have intersecting propagator edges (under the usual condition that propagator edges are all drawn on the interior of the polygon). Note that we place no restriction on the number of propagators incident upon an edge of the boundary polygon. If two propagators share a boundary edge, by the last condition of definition \ref{Wilsondef}, the other endpoints of those two propagators must be different. Furthermore, by definition \ref{crossingpropagatordef}, these are non-crossing propagators.

\subsubsection{Wilson loops as matroids \label{thematroid}}

We now define the matroid associated to a Wilson loop.

Write $M(W)$ to be the matroid assoicated to $W = (\fP, [n])$. This is the matroid realized by the matrix $\cM(W(\cZ_*))$ in Section \ref{physics}.

We use the structure of $\cM(W(\cZ_*))$ as outlined in Theorem \ref{genericindep} and Corollary \ref{subdiagramcor} to define the bases of $\cM(W(\cZ_*))$.

We begin with a definition.

\begin{dfn}
Given a set of vertices $V \in [n]$, we say that $\Prop(V) \in \fP$ is the set of propagators that depend on $V$.
\end{dfn}

Under this notation, $V \subset V_{\Prop(V)}$.

\begin{lem}
Let $V$ be a set of columns of $\cM(W(\cZ_*))$. The set $V$ is linearly independent if and only if there does not exist a subset $U \subsetneq V$ such that $|U| > |\Prop (U)|$.
\end{lem}

\begin{proof}
Recall, from Definition \ref{Wilsondef}, that we only consider well defined Wilson loop diagrams. Therefore, by Theorem \ref{genericindep}, $\cM(W(\cZ_*)$ has full rank, as do all matrices associated to subdiagrams.

Let $\fB$ be the set of sets colums of $\cM(W(\cZ_*)$ defining maximal minors with non-zero determinant.

A subset $V$ of columns of $\cM(W(\cZ_*))$ is linearly independent if and only if it is contained in a set $B \in \fB$.  Since, $\cM(W(\cZ_*))$ has full rank, $|B| = |\Prop B|$. In particular, there does not exist $U \subset B$ such that $|U| > |\Prop(U)|$, as this would imply that the matrix $\cM(W\Prop(U)(\cZ_*))$ does not have full rank, contradicting Corollary \ref{subdiagramcor}. Therefore, any such $U$ is linearly dependent, as is any set $V$ containing it.
\end{proof}

This gives a prescription for defining the basis sets of the matroids realized by matrices of the form $\cM(W(\cZ_*))$.

\begin{thm}
The matroid $M(W)$ assoiciated to the well defined Wilson loop $W = (\fP, [n])$, and realized by $\cM(W(\cZ_*)$ is defined $M(W) = ([n], \fB)$ with \ba \fB = \{B\subset [n] | |B| = |\fP| ; \; \not \exists U \subset B \textrm{ such that } |U| > |\Prop (U)|\} \label{basisset} \;. \ea
\label{correctmatroid} \end{thm}

\begin{rmk}
Since, at this juncture, we are interested in whether or not the $\cM(W(\cZ_*))$ is a positive Grassmannian for a generic twistor configuration, we drop the explicit dependence on the set of twistors defining $\cZ_*$ when referring to matroids.
\end{rmk}

Since the vertices of the polygon of $W$ correspond to the ground set of $M(W)$, we identify the two sets. In particular, if $V \subset [n]$ is a dependent set (resp. independent set, basis set, cyclic set, etc.) then we say that the corresponding vertices of the Wilson loop $W = (\fP, [n])$  are a  dependent set (resp. independent set, basis set, cyclic set, etc.).

As a direct corrollary of Theorem \ref{correctmatroid} and Corrollary \ref{subdiagramcor}, we may classify independent sets of $W$ by their propagator set.

\begin{cor}
If $B$ is an independent set of $W$, with $P = \Prop(B)$, then $B$ is independent on $W|P$.  \label{basissetcor}
\end{cor}

\begin{proof}
By equation \eqref{basisset}, $B$ is independent in $W$ if there is not a subset $S \subset B$ such that $|S| > |\Prop(S)|$. Since $S \subset B$, $\Prop (S) \subset P$. Therefore, if $B$ is independent on $W$, it is independent on $W|P$.
\end{proof}

We may go further and this, to see how the basis sets of Wilson loop diagrams vary as one adds and subtracts propagators. For any sub-Wilson loop $W|P$, the matroid $M(W|P)$ has rank $\rk M(W|P) = |P|$. This gives the following result.

\begin{prop}
Consider $W = (P, V_P)$, and $W'$ a well defined Wilson loop diagram containing $W$. Any set $B$ is independent in $W$ if and only if $|B| \leq |P|$, and $B$ is independent all subsdiagrams of $W'$ containin $W$ as a sub-Wilson loop diagram. Furthermore, it is a basis if and only if, in addition, $|B| = |P|$. \label{indepincreasecor}
\end{prop}

\begin{proof}
Write $W' =(P', [n'])$. If $W$ is a sub-Wilson loop diagram of $W$, by definition \ref{subdiagramdfn}, $P \subset P'$ and  $V_P\subset [n']$.

By Theorem \ref{correctmatroid}, $B$ is independent in $W = (P, [n])$, if and only if $|B| \leq |P|$, and it does not contain a subset $V \subsetneq B$ such that $|V| > |\Prop (V)|$. Consider any Wilson loop diagram, $U$, that is a subdiagram of $W'$ and such that $W$ is a subdiagram of $U$. Then the set $\Prop (V)$ in $W$ is contained in the set $\Prop(V)$ in $U$. Therefore, $B$ is independent in $U$.

If $W$ is a submatroid of $W'$, then the propagator structure of $W$ does not change, and the rank increases. The set $B$ is independent in $W'$ if and only if there is a basis set $B'$ of $W'$ containing it. Since $|B'| > |B|$, and there is not a subset $U \subsetneq B$ such that $|U| > |\Prop (U)|$, $B$ is independent.

If $B$ is an independent set of every sub diagram $U = (P', [n'])$ of $W'$, such that $\Prop(B) \subset P'$ and $V_{\Prop(B)} \subset [n']$, then $B$ is contained in a basis, $B'$, of $W'$. In particular, there is not a subset $U \subsetneq B$ such that $|U| > |\Prop (U)|$. Since $|B| \leq |\Prop(B)|$, by hypothesis, $B$ is independent in $W = (\Prop(B), V_{\Prop(B)})$, and a basis if and only if, in addition, $|B| = |\Prop(B)|$
\end{proof}

Thus we have established a set of conditions mediating the independence of vertices among subdiagrams. However, showing that a set of vertices is independent in all subdiagrams is onerous. Next, we show that one may always find some subdiagram for which an independents set is a basis.

Thus we have given a condition such that that a independent set, $B$ on a Wilson loop diagram restricts to an independent set on the subdiagram $W | \Prop(B)$.  Next, we show that, with possible further restriction of the diagram, one may always find a basis.

\begin{prop}
If $B$ is an independent set in $W = (\fP, [n])$, then there exists a $P \subset \Prop(B)$ such that $|P| = |B|$, and $B$ is a basis of $W|P$.
\label{indeppropagatorprop}\end{prop}

\begin{proof}
Suppose $B$ is independent in $W$. by Theorem \ref{correctmatroid}, $B$ does not contain any subset $S \subset B$ such that $|S| < |\Prop (S)|$. For any subset $S \in B$ $\Prop(S) \subset \Prop(B)$. Therefore, $B$ is an independent set in $W | \Prop(B)$.

If $|B| = |\Prop(B)|$, then $P = \Prop(B)$, and $B$ is a basis, again by Theorem \ref{correctmatroid}. It remains to consider when $|B| < |\Prop(B)|$.

Suppose, without loss of generality, that there is no subset $S \subset B$ such that $|S| = |\Prop S|$. If there is such a set, $S$  is a basis of $W | S$. Consider the independents set $B\setminus S$ in $W|(\Prop(B) \setminus \Prop(S))$.

Then for any $p \in B$ such that $|V_p \cap B| \neq 1$, $|B| \leq |\Prop(B) \setminus p|$, and every subset, $S \subset B$ has the property that $|S| \leq |\Prop (S) \setminus p|$ in $W| (\Prop(B) \setminus p)$.

We may proceed to remove propagators and subbasis of $B$ in this fashion until $|B| < |\Prop(B) \setminus Q|$, for some set of propagators $Q \subset \Prop(B)$. Since $B$ is independent, it is a basis of $W | (\Prop(B) \setminus Q)$.
\end{proof}

Furthermore, we see that for any set of vertices $V$ in $W$, we have \ba \rk V \leq \min \{|V|, |\Prop(V)|\} \label{rankbound} \;. \ea Indeed, by definition, $\rk (V)  = \max \{|V \cap B| | B \textrm{ is a basis set of }W\}$. Therefore, $\rk(V) \leq |V|$. By Proposition \ref{indeppropagatorprop}, we see that the $\rk(V)$ in $W$ is the same as $\rk(V)$ in $\Prop(V)$.

If there are more propagators than vertices in a set $V$, that is, if $|\Prop(V)| \leq |V|$, then by Corrollary \ref{basissetcor}, $\rk (V)$ in $W$ is the same as $\rk (V)$ in $W|\Prop(V)$. Therefore, $\rk (V) \leq |\Prop(V)|$.

Finally, we give an example of the matrices associated to Wilson loop diagrams and the associated matroid structures.

\begin{eg}\label{Wilsontomatroideg}
As developed in Section \ref{physics}, the matrix $\cM(W(\cZ_*))$ is defined as follows: \ba \cM(W)_{ij} = \begin{cases} 0 & \textrm{ if } j \not \in V_{p_i} \\ c_{ij} & \textrm{else; } j \in V_{p_i}\end{cases} \; , \label{realizationmatix}\ea where each $c_{ij}$ is a function on the twistor configuration $\cZ_*^\mu$.

Consider the Wilson loop diagram presented in Example \ref{disconnected3eg}. Write \bas W = (\{(2, 4),(4, 7), (5,7)\},[8]) =
{\begin{xy}
(-5,10) *{\bullet},
(-6, 10); (6, 10) **{\dir{-}},
(-4, 11); (-11, 4) **{\dir{-}},
(-10, 6); (-10, -6) **{\dir{-}},
(-11, -4); (-4, -11) **{\dir{-}},
(-6, -10); (6, -10) **{\dir{-}},
(4, 11); (11, 4) **{\dir{-}},
(10, 6); (10, -6) **{\dir{-}},
(11, -4); (4, -11) **{\dir{-}},
(-10, 0); (0, -10) **@{~},
(2, -10); (7, 8) **@{~},
(8, -7); (8, 7) **@{~},
\end{xy}} \;, \eas as above.
 Then, for generic $\cZ_*$,
\bas M(W(\cZ_*)) = \left(
\begin{array}{cccccccc}
0 & c_{1,2} & c_{1,3} & c_{1,4} & c_{1,5} & 0 & 0 & 0 \\
0 & 0 & 0 & c_{2,4} & c_{2,5} & 0 & c_{2,7} & c_{2,8} \\
0 & 0 & 0 & 0 & c_{3,5} & c_{3,6} & c_{2,7} & c_{3,8} \\
\end{array}
\right) \eas

The basis of $M(W)$ are
\bmls \fB =(\{(2,4,5), (2,4,6), (2,4,7), (2, 4, 7), (2, 5, 6), (2, 5, 7), (2, 6, 7), (2, 7, 8), \\ (4, 5, 6), (4, 5, 7), (5, 6, 7), \textrm{ and all distinct sets formed from these by substituting} \\ \textrm{3 for 2, and 8 for 7}\} ) \;.\emls
\end{eg}

\begin{rmk} In section \ref{Wilsonamplitudes}, we discussed that in this paper, we relax the physical condition that, for a Wilson loop diagram $W = (\fP, [n])$,  $n \geq |V_{\fP}| + 4$, to allow the combinatorial studying of Wilson loop diagrams. In the combinatorial approach, we are not directly interested in $n$ particle $N^kMHV$ diagrams, but in the subdiagrams that determine it. In particular, we are concerned with the number of particles that interact independently with each other, or subsets of propagators with distinct dependency sets, as this is what defines the matroid $M(W)$. For instance, in the diagram $W(\{(e_2, e_4),(e_4, e_7), (e_5,e_7)\},8)$ in Example \ref{disconnected3eg} none of the propagators depend on the first vertex. Therefore, the matrix $\cM(W)$ has $7$ non-zero columns. Similarly, in the diagram \bas W = (\{(e_2, e_6),(e_4, e_8)\},8) =
{\begin{xy}
(-5,10) *{\bullet},
(-6, 10); (6, 10) **{\dir{-}},
(-4, 11); (-11, 4) **{\dir{-}},
(-10, 6); (-10, -6) **{\dir{-}},
(-11, -4); (-4, -11) **{\dir{-}},
(-6, -10); (6, -10) **{\dir{-}},
(4, 11); (11, 4) **{\dir{-}},
(10, 6); (10, -6) **{\dir{-}},
(11, -4); (4, -11) **{\dir{-}},
(-8, 7); (8, 7) **@{~},
(-8, -7); (8, -7) **@{~},
\end{xy}} \;, \eas one propagator has the dependency set  $V_{p_1} = \{1, 2, 7, 8\}$, while the other has the dependency set $V_{p_2} =\{3, 4, 5, 6\}$. In other words, the two propagators are independent. The matroid $M(W)$ can be written as a direct sum of two matroids, $M_1$ and $M_2$ depending on $V_{p_1}$ and $V_{p_2}$ respectively. That is, $M(W)$ is disconnected. We may write $\cM_1 = \cM(W_1)$ and $\cM_2 = \cM(W_2)$, for $W_1$ and $W_2$ two different exact Wilson loops with one propagator and four vertices. These diagrams uniquely determine the interaction above, but neither of these sub-Wilson loops do not correspond to directly to $n$ particle $N^kMHV$ diagrams.
\end{rmk}

\subsubsection{Circuits and flats of Wilson loops\label{circuitsandflats}}

In this section, we define circuits of $M(W)$, and an important family of flats.

\begin{thm}
A set $C \subset [n]$ is a circuit of $M(W)$ if and only if \begin{enumerate}
\item $|C| = \rk (C) +1$
\item There does not exist a proper subset $S \subsetneq C$ such that $|S| > |\Prop (S)|$.
\end{enumerate}
\label{circuitcor}\end{thm}

\begin{proof}
The condition on the size of $C$ comes from the definition of a circuit. The second condition comes from the fact that a circuit is a minimally dependent set. That is, there are no proper subsets that are dependent sets, which is equivalent, by equation \eqref{basisset}, to saying that there does not exists a proper subset $S \subsetneq C$ such that $|S| > |\Prop (S)|$.
\end{proof}

This, in conjuction with equation \eqref{rankbound}, gives a relationship between the rank of a circuit and its propagator set.

\begin{cor}\label{circuitrank}
A set $C$ is a circuit if and only if $\rk (C) = |\Prop(C)|$.
\end{cor}

\begin{proof}
By the first condition of corrollary \ref{circuitcor}, and equation \eqref{rankbound}, we see that $\rk (C) \leq |\Prop(C)|$. If $\rk (C) < |\Prop(C)|$, then by Theorem \ref{correctmatroid}, there exists a proper subset $S \subsetneq C$ such that $|S| > |\Prop (S)|$, and $C$ is not a circuit.
\end{proof}

Finally, we define an important family of flats that are easy to see from the diagramatics.

\begin{dfn}
Write $F(P)\subset [n]$ to indicate the set of vertices in $[n]$ that only define propagators in $P$.
\end{dfn}

We may write $F(P) = V_P \setminus V_{P^c}$. By this definition, it is clear that if $Q \subset P$, then $F(Q) \subset F(P)$.

In this notation, $F(\emptyset)$ is the set of vertices of $W$ that are not adjacent to any propagators. There is a relationship between flats and dependency sets. One sees from the definition of $F(P)$ that $F(P)^c = V_{P^c}$. However, it is not true that, for any set of propagators, $P \subset \fP$, $V_P^c = F(P^c)$. In general, $P^c$ is too large a set. For any set of propagators $P \subset \fP$, one may write $V_P^c = F(Q)$, for some subset $Q \subset P^c$ such that $V_Q = V_{P^c}$. In this case, we have $V_P = F(Q)^c$.

\begin{lem}
Consider $W = (\fP, [n])$. For any $P \subset \fP$, the set $F(P) \cup F(\emptyset)$ is a flat in $M(W)$. \label{propagatorflat}
\end{lem}

\begin{proof}
By construction, $\Prop(F(P)) = P$. If $e \not \in F(P) \cup F(\emptyset)$, then \bas \Prop(F(P) \cup e)) = P \cup \Prop(e) \supsetneq P\;. \eas

By definiton of the rank function, there exists a basis of $M(W)$, $B$ such that $|B \cap F(P)| = \rk(F(P))$. Furthermore, by Proposition \ref{indeppropagatorprop}, $\rk(F(P))$  in $W|P$ is the same as $\rk(F(P))$ in $W$. Therefore, any such $B$ that maximally intersects $B(P)$ contains a basis of $W|P$. Thus $\Prop(B \cap F(P)) = P$.

Consider the set $S = e \cup (B \cap F(P))$. For $e \not \in F(P) \cup F(\emptyset)$, $\Prop(e) \not \subset P$/ Therefore, we write $\Prop S = P \cup \Prop(e)$. By equation \eqref{rankbound}, $\rk(F(P))  = |B \cap F(P)| \leq |P|$. Therefore, \bas |S| = |B \cap F(P)| + 1 \leq |P| + 1 \; .\eas Since $|P| < |\Prop(S)|$, we have that  $|S| \leq |P| + 1 \leq |\Prop(S)|$, making $S$ an independent set of size greater than $\rk(F(P))$.

Therefore, $F(P) \cup F(\emptyset)$ is a flat in $M(W)$.
\end{proof}

As a corrollary to Lemma \ref{propagatorflat}, we see that the set $F(P)$ defines a flat of the Wilson loop $W = (\fP, V_{\fP})$. This is because, by construction, $F(\emptyset)= \emptyset$ in $(\fP, V_{\fP})$. This motivates the following definition.

\begin{dfn}
The set $F(P)$ is called the propagator flat of $(\fP, [n])$, defined by $P \subset \fP)$.
\end{dfn}

Propagator flats play an imporatnat role in determining which Wilson loop diagrams lead to positorids.

\subsection{Wilson loop diagrams as positroids\label{Wilsonpositroid}}

The aim of this paper is to identify admissible Wilson loops. From Definition \ref{admissibledfn}, this means we only wish to consider Wilson loops such that $\cM(W(\cZ_*))$ realizes a non-negative Grassmannian, i.e. define a positroid. In this section, we determine, graphically, which Wilson loop diagrams are are admissible by studying which matroids $M(W)$ are positroids.

From Theorem \ref{positroidcondition}, the connected matriod $M(W)$ is a positroid if an only if all flacets are cyclic intervals. Therefore, it is necessary to turn to Theorem \ref{flacetcondition}, which defines flacets, to determine which sets of vertices of a Wilson loop define flacets. This section is devoted to encoding the conditions laid out in Theorem \ref{flacetcondition} in terms of graphical properties of the Wilson loop.

\subsubsection{Connected Wilson loops}

As shown in Section \ref{matroiddefssection}, connected matroids are the building blocks of matroid theory. Following the philosophy laid out in section \ref{thematroid}, we continue to identify the Wilson loop diagram with its matroid. Therefore, we define connected Wilson loops to be exactly those that give rise to connected matroids.

For the rest of this paper, we restrict ourselves to working only with connected Wilson loops. Theorem \ref{disconnpositroid} and gives a prescription for piecing together connected positroids into larger, disconnected positroids. In this section, we extend this to Wilson loop diagrams. Definition \ref{noncrossingpartitiondfn} lends itself naturally to the Wilson loop diagram setting. We use this reinterpertation in Theorem \ref{noncrosspartconn} to find a prescription to determine which disconnected Wilson loop diagrams are also positroids, and therefore admissible.

We begin by defining connected Wilson loops.

\begin{dfn}
A Wilson loop diagram is called connected if $M(W)$ is connected. Otherwise, it is disconnected. \end{dfn}

Notice that the diagram in Example \ref{disconnected3eg} is disconnected. A matroid $M$ is disconnected if it can be split into two independent sets of data. Similarly, a Wilson loop is disconnected if its propagators can be split into disjoint sets such that their dependency sets are also disjoint. We show this below.

\begin{lem} If $F(\emptyset) \neq \emptyset$ in $W = (\fP, [n])$, then $M(W)$ is disconnected. \label{emptyflatprop}\end{lem}

\begin{proof}

If $F(\emptyset) \neq \emptyset$, then $M(W)$ is disconnected. Write $M = M|F(\emptyset) \oplus M|([n]\setminus F(\emptyset))$.
\end{proof}

Therefore, if a Wilson loop has vertices that do not define propagators, it is a disconnected loop. Next we consider Wilson loops where that it not the case, i.e Wilson loops for which $F(\emptyset) = \emptyset$.

\begin{thm} \label{disconnectedloopprop}
A Wilson loop diagram $W(\fP,V_{\fP})$ is disconnected if and only if $\fP$ can be partitioned two sets, $\fP =  P_1 \amalg P_2$, such that their propagator flats $F(P_1)$, $F(P_2)$ partition $[n]$.
\end{thm}

\begin{proof}
By Definition \ref{matroiddefs}, we see that the matroid $M(W)$ is disconnected if and only if $V_{\fP}$ can be partitoned into flats. Therefore, if such a partition of propagator flats exists, $M(W)$ is disconnected.

We prove the converse by induction on number of propagators. Suppose $W = (P, V_P)$ where $P$ has only one propagator. The matroid $M(W)$ is connected, by direct calculation. Furthermore, it is impossible to to partition $P$ as above. Suppose $W= (P, V_P)$ where $P$ has two propagators. If the matroid $M(W)$ is disconnected, then there exists flats $F_1$ and $F_2$ each of rank one that are mutually independent. This implies that, for $P = \{p_1, p_2\}$, $F_1 = V_{p_1} $ and $F_2 = V_{p_2}$.

Suppose, the theorem holds for all Wilson loop diagrams of the form $W = (\fP, V_{\fP})$ with $|\fP| < k$ propagators. Consider any disconnected Wilson loop diagram $W = (\fP, V_{\fP})$ with $|\fP| = k$ such that there does not exists a partition of $\fP$, $P_1 \amalg P_2 = \fP$, with $M(W) = M(W)|_{F(P_1)} \oplus M(W)|_{F(P_2)}$. Then write $M(W) = M(W)|_{F_1} \oplus M(W)|_{F_2}$, where $F_1$ and $F_2$ are two flats of $W$ that partition $V(P)$. Write $P_1 = \Prop(F_1)$, $P_2 = \Prop(F_2)$, and $P = P_1 \cap P_2$. Since $F_1$ and $F_2$ are both flats, and $\rk M(W) = k$, $|P| <k$. By induction, the Wilson loop diagram $(P, V_P)$ is not connected, since, by hypothesis, the set $P$ cannot be partitioned to form a partition of the set $V_P$.

The diagram $(P, V_P)$ is a subdiagram of $W$. If the set $B$, a basis of $(P, V_P)$ cannot be divided into mutually independent subsets, then, by Proposition \ref{indepincreasecor}, neither can $B$ as an independent set of $M(W)|V_P$. Therefore, $M(W)|V_P$ is not disconnected. Since the set $V_P$ intersects both flats, $F_1 \cap V_P \neq \emptyset$, and $F_2 \cap V_P \neq \emptyset$, this implies that $F_1$ and $F_2$ are not mutually independent. Therefore, $M(W)$ is not disconnected.
\end{proof}

We proceed with examples of disconnected Wilson loop diagrams.

\begin{eg} \label{decomposedeg}
Here are two examples of disconnected Wilson loop diagrams. The first is the diagram from Example \ref{disconnected3eg}
\bas W = (\{(e_2, e_4),(e_4, e_7), (e_5,e_7)\},[8]) =
{\begin{xy}
(-5,10) *{\bullet},
(-6, 10); (6, 10) **{\dir{-}},
(-4, 11); (-11, 4) **{\dir{-}},
(-10, 6); (-10, -6) **{\dir{-}},
(-11, -4); (-4, -11) **{\dir{-}},
(-6, -10); (6, -10) **{\dir{-}},
(4, 11); (11, 4) **{\dir{-}},
(10, 6); (10, -6) **{\dir{-}},
(11, -4); (4, -11) **{\dir{-}},
(-10, 0); (0, -10) **@{~},
(2, -10); (7, 8) **@{~},
(8, -7); (8, 7) **@{~},
\end{xy}} \;. \eas In this case, $[n]$ can be partitioned into $F(\emptyset) = [1]$ and $F(\fP) = [2,8]$. Consisder another graph,
\bas (\{(e_2, e_6),(e_4, e_8)\},[8]) =
{\begin{xy}
(-5,10) *{\bullet},
(-6, 10); (6, 10) **{\dir{-}},
(-4, 11); (-11, 4) **{\dir{-}},
(-10, 6); (-10, -6) **{\dir{-}},
(-11, -4); (-4, -11) **{\dir{-}},
(-6, -10); (6, -10) **{\dir{-}},
(4, 11); (11, 4) **{\dir{-}},
(10, 6); (10, -6) **{\dir{-}},
(11, -4); (4, -11) **{\dir{-}},
(-10, 0); (10, 0) **@{~},
(0, -10); (0, 10) **@{~},
\end{xy}} \;. \eas In this case, write $\fP = \{p_h, p_v\}$, with $h$ and $v$ for horizontal and vertical. Then $F(p_h) = [2,3] \cup [6,7]$ and $F(p_v) = [4,5] \cup [8,1]$.

In the first graph, by identifying the cyclic set $[2,8]$ with $[7]$, we may say that the connected components of $W$ are $(\emptyset, [1])$ and $(\fP, [7])$, where \bas (\emptyset, [1]) = {{\begin{xy}
(0, 0) *{\bullet}= "a",
"a"; "a" **\crv{+(-5,3)&+(5,7)&+(5,-7)},
\end{xy}}} \quad ; \quad W(\fP, [7]) = {\begin{xy}
(-10, 4.5) *{\bullet},
(-11, 4); (1, 11) **{\dir{-}},
(-10, 6); (-10, -6) **{\dir{-}},
(-11, -4); (-4, -11) **{\dir{-}},
(-6, -10); (6, -10) **{\dir{-}},
(-1, 11); (11, 4) **{\dir{-}},
(10, 6); (10, -6) **{\dir{-}},
(11, -4); (4, -11) **{\dir{-}},
(-10, 0); (0, -10) **@{~},
(2, -10); (2, 9) **@{~},
(8, -7); (8, 6) **@{~},
\end{xy}} \eas

On the level of matroids, it is easy to see that the corresponding matroid, as defined in Example \ref{Wilsontomatroideg}, is disconnected. Namely, it can be realized as \bas \cM(W(2,4, 4, 7, 5, 7)(8)) = \left(
\begin{array}{cccccccc}
0 & x_{1,2} & x_{1,3} & x_{1,4} & x_{1,5} & 0 & 0 & 0 \\
0 & 0 & 0 & x_{2,4} & x_{2,5} & 0 & x_{2,6} & x_{2,7} \\
0 & 0 & 0 & 0 & x_{3,5} & x_{3,6} & x_{3,7} & x_{3,8} \\
\end{array}
\right) \eas which can be written \bmls \cM(W(2,4, 4, 7, 5, 7)(8)) = \left(\begin{array}{c} 0\\ 0\\0 \end{array} \right) \oplus \left(
\begin{array}{ccccccc}
x_{1,2} & x_{1,3} & x_{1,4} & x_{1,5} & 0 & 0 & 0 \\
0 & 0 & x_{2,4} & x_{2,5} & 0 & x_{2,6} & x_{2,7} \\
0 & 0 & 0 & x_{3,5} & x_{3,6} & x_{3,7} & x_{3,8} \\
\end{array}
\right) \\ = \cM((\emptyset, [1])(Z_*, Z_1)) \oplus \cM(W(\fP, [7])(\cZ_*|V_{\fP}^*)) \;.\emls
\end{eg}

There is a corollary to Proposition \ref{disconnectedloopprop} that is easier to state.

\begin{cor}
A Wilson loop diagram, $W=(\fP,[n])$ is disconnected if and only if there is a subset $P \subseteq \fP$ such that $F(P) = V_P \subsetneq [n]$. Equivalently, if $P = \Prop(V_P)$.
\label{disconnectedloopcor}\end{cor}

\begin{proof}
If $F(\emptyset) \neq \emptyset$ in $W=(\fP,n)$, then $M(W)$ is disconnected. This is equivalent to the statement that $F(\fP) = V_{\fP}$, and $F(\fP) \amalg F(\emptyset) = [n]$.

We consider Wilson loops such that $F(\emptyset) = \emptyset$. By construction of $F(P)$, $\Prop(F(P)) = P$. Thereofore $F(P) = V_P$ if and only if $P = \Prop(V_P)$

The flat defined by a set of propagators $P$, $F(P)$ is the set of vertices that are adjacent only to those propagators. This is the entire set of vertices adjacent to those propagators if and only if there are no other propagators adjacent to those vertices, i.e. if $P = \Prop(V_P)$.

The required $P \subset \fP$ such that $F(P) = V_P$ exists if and only if $F(P^c) = V_{P^c}$. Therefore $P, P^c$ form the partition of $\fP$ required by Proposition \ref{disconnectedloopprop}.
\end{proof}

We can now combine this classification of connected Wilson loops diagrams with Theorem \ref{disconnpositroid} to give a classification of which disconnected Wilson loop graphs with respect to positroids.

\begin{thm} \label{noncrosspartconn}
Let $W = (\fP, [n])$ be a disconnected Wilson loop diagram, with a partition \bas [n] = (\amalg_{i=1}^k F(P_i))\amalg F(\emptyset) \;.\eas The Wilson loop $W$ is admissible if an only if each connected component, $W_i = (P_i, F(P_i))$, (such that $P_i \neq \emptyset$), is admissible, and the sets of propagators $P_i$ do not cross each other.
\end{thm}

\begin{proof}
This theorem follows from Theorem \ref{disconnpositroid}. The Wilson loop diagram $W = (\fP, n)$ is admissible if and only if $M(W)$ is a positroid. If the disconnected matroid $M(W)$ is a positroid, by Theorem \ref{noncrosspartconn} the ground sets for each connected component forms a non-crossing partition of $[n]$. Since $W$ is disconnected, by hypothesis, and by Proposition \ref{disconnectedloopprop}, $\{F(P_i) \cup F(\emptyset)\}$ partition $[n]$. It remains to check that a crossing partition on $[n]$ is equivalent to sets of crossing propagators.

The matroid $M|F(\emptyset)$ can be written as the direct sum of matroids of rank $0$, with ground sets consisting of one element \bas M|F(\emptyset) = \bigoplus_{v \in F(\emptyset)} (v, \fB =\emptyset) \;.\eas By definition \ref{noncrossingpartitiondfn}, ground sets of $1$ element cannot be a part of a non-crossing partition. Therefore, we may ignore $F(\emptyset)$, or, equivalently, assume $F(\emptyset) = \emptyset$.

If $\{F(P_1) \ldots F(P_k)\}$ is a non-crossing partition of $[n]$, by Definition \ref{noncrossingpartitiondfn}, the vertices of $F(P_i)$ define non-crossing polygons by connecting the vertices of each partition set. By Corrollary \ref{disconnectedloopcor}, $F(P_i) = V(P_i)$. Therefore, these polygons fully contain all propagators in the set $P_i$. Therefore, the polygons on $F(P_i)$ are exactly the boundary polygons of the connected components of $W$, $W_i$.

Since each $W_i$ is connected, and their vertex sets disjoint, the polygons $W_i$ intersect if and only if the propagators cross. Therefore, $F(P_i)$ form a non-crossing partition of $[n]$ if and only if the sets of propagators do not cross.
\end{proof}

We have yet to graphically define which connected Wilson loops are admissible. However, this theorem shows that if the propagators of one connected admissible Wilson loop subdiagram intersect the propagators of another admissible Wilson loop subdiagram, then the composite diagram is not admissible.

\begin{eg}
The disconnected diagram from Example \ref{decomposedeg}
\bas (\{(e_2, e_6),(e_4, e_8)\},8) =
{\begin{xy}
(-5,10) *{\bullet},
(-6, 10); (6, 10) **{\dir{-}},
(-4, 11); (-11, 4) **{\dir{-}},
(-10, 6); (-10, -6) **{\dir{-}},
(-11, -4); (-4, -11) **{\dir{-}},
(-6, -10); (6, -10) **{\dir{-}},
(4, 11); (11, 4) **{\dir{-}},
(10, 6); (10, -6) **{\dir{-}},
(11, -4); (4, -11) **{\dir{-}},
(-10, 0); (10, 0) **@{~},
(0, -10); (0, 10) **@{~},
\end{xy}}  \eas is not admissible.

By inspection, one sees that composite diagrams $\begin{xy}
(-7,7) *{\bullet},
(-8, 7); (8, 7) **{\dir{-}},
(-7, 8); (-7, -8) **{\dir{-}},
(-8, -7); (8, -7) **{\dir{-}},
(7, 8); (7, -8) **{\dir{-}},
(-7, 0); (7, 0) **@{~},
\end{xy}$ and $\begin{xy}
(-7,7) *{\bullet},
(-8, 7); (8, 7) **{\dir{-}},
(-7, 8); (-7, -8) **{\dir{-}},
(-11, -7); (8, -7) **{\dir{-}},
(7, 8); (7, -8) **{\dir{-}},
(0, -7); (0, 7) **@{~},
\end{xy}$ are both admissible. However, since the propagators cross, this diagram is not admissible.
\end{eg}

\begin{rmk}
Henceforth, we only consider connected Wilson loops unless specifically stated otherwise.
\end{rmk}

\subsubsection{Cyclic flats for Wilson loops\label{Wilsoncyclicflats}}

At the end of the day, we are interested in flacets of connected Wilson loops, which define non-negative Grassmannians. By Theorem \ref{flacetcondition}, flacets are defined by cyclic flats with certain conditions imposed upon them. Section \ref{Wilsonflacets} defines the flacets. In this section, we show that cyclic flats are propagator flats, $F(P)$, satisfying $|F(P)| \geq |P|$.

\begin{dfn}
Let $F$ be a flat of $M(W)$. We say $F$ is a propagator flat if $F = F(P)$, for some subset of propagators $P \in \fP$.
\end{dfn}

Recall that a cyclic flat is one that can be written as a union of circuits. We begin with a result on the rank of cyclic sets.

\begin{lem} \label{cyclicrank}
If $C$ is a cyclic set in $M(W)$, then $\rk (C) = |\Prop(C)|$.
\end{lem}

\begin{proof}
Since $C$ is a cyclic set, write $C = \cup C_i$, where each $C_i$ is a circuit. By Corollary \ref{circuitrank}, $\rk C_i = |\Prop (C_i)|$. By equation \eqref{rankbound}, $\rk (C) \leq |\Prop(C)|$.

If $\rk (C) < |\Prop(C)|$, by Proposition \ref{indeppropagatorprop} there exists a subset $Q \subset \Prop(C)$ such that the rank of $C$ is the same in the matroid $M(W|\Prop(C))$ as in the matroid $M(W|Q)$. Consider a circuit $C_i$ such that $Q \cap \Prop(C_i) \neq \emptyset$. Again by proposition \ref{indeppropagatorprop}, the rank of $C_i$ is the same in the matroid $M(W|\Prop(C_i))$ as in the matroid $M(W|\Prop(C_i)\setminus Q)$. But by Corollary \ref{circuitrank}, $\rk C_i = |\Prop C_i|$. Therefore, any set $C_i$ properly intersectin $Q$ is not a circuit. Therefore, $C$ is not a cyclic set.
\end{proof}

Finally we show that if a flat is not a propagator flat, such that $|F(P)| > |P|$, then it is not cyclic.

\begin{lem} Let $F$ be a cyclic flat. Then there is a $P \subset \fP$, such that $F = F(P)$, and $|F| > |P|$.  \label{nonpropnoncyclic}\end{lem}
\begin{proof}
If $|F| < |\Prop F|$, then, by equation \eqref{rankbound}, $\rk (F) \leq |F| < |\Prop F|$. By Lemma \ref{cyclicrank}, $F$ is not cyclic.

If $|F| = |\Prop F|$, there are two possible cases. The first is that $\rk (F) =  |F|$, which implies that $F$ is an independent set. By definition, $F$ is not cyclic. Otherwise $\rk (F) \leq |\Prop(F)|$. Again, by Lemma \ref{cyclicrank}, $F$ is not cyclic.

Therefore, assume that $|F| > |\Prop(F)|$.

For any $V \in [n]$, write  $\textrm{cl}(V)$ to indicate the smallest flat of $M(W)$ containing $V$. In standard matroidal language, this is the closure of the set $V$. Explicitly, \bas \textrm{cl}(V) = \{e \in [n] | \rk (V \cup e) = \rk (V) \} \;. \eas

Suppose there does not exists a set $P \subset \fP$, such that $F = F(P)$. Then $\textrm{cl}(V) \subsetneq F(\Prop (V))$. In particular, there exists $e \in F(\Prop (V)) \setminus \textrm{cl}(V)$. By definition of $\textrm{cl}(V)$, \bas \rk (\textrm{cl}(V) \cup e) = \rk (\textrm{cl}(V)) + 1 \;.\eas By equation \eqref{rankbound}, \bas  \rk (\textrm{cl}(V)) + 1 \leq |\Prop(V)|\;.\eas Therefore, $\rk (\textrm{cl}(V)) < |\Prop(V)|$. By Theorem \ref{cyclicrank}, $\textrm{cl}(V))$ is not a cyclic flat.
\end{proof}

\subsubsection{Flacets for Wilson loops\label{Wilsonflacets}}
We are now ready to define which flats of $M(W)$ are flacets. By Theorem \ref{flacetcondition} and Lemma \ref{cyclicrank}, we only need to consider propagator flats such that $|F(P)| < |P|$, and $\rk F(P) = |P|$ to identify the flacets of $M(W)$.

\begin{dfn}
If $F$ is a propagator flat, with $F = F(P)$, $|F(P)| < |P|$, and $\rk F(P) = |P|$, we say that $F$ is a propagator flat of maximal rank.
\end{dfn}

In this section, we finally identify the flacets of a  Wilson loop diagram. We find that the flacets depend purely on the propagator structure of the Wilson loop diagrams. Namely, we identify which propagators flats of maximal rank satisfy the conditions that $M(W)|_{F(P)}$ and $M(W)/F(P)$ are connected.

We begin by studying the restriction. Unlike before, where we have identified the Wilson loop diagram to its matroid, we must make a distinction when studying restrictions. This is because the restriction of a connected matroid $M(W)$ is not necessarily a matroid associated to a Wilson loop diagram.

Also, the notation for restriction on Wilson loop diagrams and their matrices do not correspond. In particular, when we write $M(W)|E$, we restrict the \emph{matroid} to a subset of vertices. However, when we write $W|P$, we restrict the \emph{diagram} to a subset of propagators. This latter notation becomes important again when studying contractions.

First, we consider the connectivity of restrictions on $M(W)$ by propagator flats of maximal rank.

\begin{lem} For any propagator flat of maximal rank, $F(P)$, if the restriction $M(W)|F(P)$ is disconnected then either one can find a partition of $F(P)$, $F_1 \amalg F_2 = F(P)$ that defines a partition of propagators, $\Prop (F_1) \amalg \Prop (F_2) = P $, or $F(P)$ is not a cyclic flat. Furthermore, if there is a patrtition of propagators $P_1 \amalg P_2 = P$, such that their flats $F_1 = F(P_1)$ and $F_1 = F(P_2)$ partition $F(P)$ then $M(W)|F(P)$ is disconnected.
\end{lem}

\begin{proof}
The first part of this lemma is a direct application of corrollary \ref{partialcyclicequiv}. If $M(W)|F(P)$ is disconnected, then either the flat $F(P)$ is not a cyclic set, or $F(P)$ can be partitioned into cyclic flats. The latter implies that $F(P)$ can be partitioned into propagator flats $F(P_1)$ and $(P_2)$. Suppose $P_1$ and $P_2$ do not partition $P$. In particular, $Q = P_1 \cap P_2$.  Then, for $i \in \{1, 2\}$, \ba |P_i \setminus Q | \leq \rk F_i \leq |P_i| \;. \label{rankineq} \ea

First we consider the case when at least one of the inequalities is an equality for each $F_i$. Indeed, since $\rk F(P) = |P|$, the first inequality of equation \eqref{rankineq} is an equality for $F_1$, say, if and only if the second inequality is an equality for $F_2$. Therefore, considering an equality for one set $F_i$ is the same as considering an equality for both sets.

The first inequality in equation \eqref{rankineq} an equality if and only if $F_i$ is an independent set on $W|(P_i \setminus Q_i)$ (in which case $F(P)$ cannot be cyclic) or if $Q = \emptyset$. Therefore, the theorem holds when, for some $F_i$, $\rk F_i = |P_i|$.

Consider the case when, neither inequalities in equation \eqref{rankineq} are equalities. Then, $\rk F_i < |P_i|$. By Lemma \ref{cyclicrank}, this implies that neither $F_1$ nor $F_2$ is cyclic. Corrollary \ref{partialcyclicequiv} then implies that $F(P)$ is not cyclic.

As a partial converse, if $P$ can be partitioned into two sets $P_1$ and $P_2$ such that $P_i = \Prop (F_i)$, each $F_i$ is the propagator flat $F(P_i)$, which are, by virtue of being flats, mutually independent. Thus $M(W)|F(P)$ is not connected.
\end{proof}

In particular, if $F(P)$, a propagator flat of maximal rank cannot be written as the union of two disjoint propagator flats, the restriction $M(W)|F(P)$ is connected.

We next consider connectivity of the contracted matroid, $M(W)/F(P)$. Unlike the restriction operator, the contraction of a matroid associated to a Wilson loop by a propagator flat of maximal rank gives another matroid associated to a Wilson loop. We begin with a definition.

\begin{dfn}
For $W = (\fP, [n])$, and $P \in \fP$, define the sub Wilson loop diagram \bas W/ P = (P^c, V_{P^c}) \;.\eas \label{loopcontractdfn}
\end{dfn}

In other words, the Wilson loop $W/P$ is formed by considering only the propagators in not in $P$, and the cyclically ordered set of vertices $F(P)^c$. We may write $M(W/P) = M(W|P^c)/F(\emptyset)$.

Before proceeding, we give an example.

\begin{eg}
Consider $W = (\fP, [7])$ as in example \ref{decomposedeg}, where \bas W = {\begin{xy}
(-10, 4.5) *{\bullet},
(-11, 4); (1, 11) **{\dir{-}},
(-10, 6); (-10, -6) **{\dir{-}},
(-11, -4); (-4, -11) **{\dir{-}},
(-6, -10); (6, -10) **{\dir{-}},
(-1, 11); (11, 4) **{\dir{-}},
(10, 6); (10, -6) **{\dir{-}},
(11, -4); (4, -11) **{\dir{-}},
(-10, 0); (0, -10) **@{~},
(2, -10); (2, 9) **@{~},
(8, -7); (8, 6) **@{~},
\end{xy}} \;. \eas As usual, write $p_1 = (1,3)$. Then $F(p_1) = \{1, 2\}$, and \bas W/\{p_1\} = 
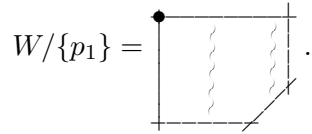
{\begin{xy}
(-7, 4) *{\bullet},
(-7, -11); (-7, 5) **{\dir{-}},
(-8, -10); (6, -10) **{\dir{-}},
(-8, 4); (11, 4) **{\dir{-}},
(10, 6); (10, -6) **{\dir{-}},
(11, -4); (4, -11) **{\dir{-}},
(0, -10); (0, 4) **@{~},
(8, -7); (8, 4) **@{~},
\end{xy}} \;.\eas
 In this case, since $F(\emptyset) = \emptyset$ in $W$, $W/ \{p_1\} = W | P^c$.
\end{eg}

If $F(P)$ is a propagator flat of maximal rank in $W_1$, then the diagram $W / P$ corresponds to a contraction of the corresponding matroid.

\begin{thm} \label{contractionmatroid}
 The equality on the level of matroids \bas M(W)/F(P) = M(W/P) \; \eas holds if and only if $F(P)$ is a propagator flat of maximal rank.
\end{thm}

\begin{proof}
Recall from definition \ref{contractionmatroid} that if $M(W) = ([n] , \fB)$, the contaction $M(W)/F = (F^c , \fB/F)$, where \bas \fB/F = \{ B \setminus F | B \in \fB, |B \cap F| \textrm{ maximal } \} \;.\eas

By Definition \ref{loopcontractdfn}, $W/P = (P^c , V_{P^c})$. By Proposition \ref{indepincreasecor}, any basis, $B_{P^c}$, of the Wilson loop diagram $W/P$ is an independent set of $W$. Since $B_{P^c}$ is a basis, we know that $|B_{P^c}| = |P^c|$. Extend $B_{P^c}$ to a basis $B$ of $W$.

Suppose that $\rk F(P) = |P|$. Since $F(P)^c = V_{P^c}$, \bas B_{P^c} = B \cap V_{P^c} = B \setminus F(P)\;. \eas Therefore $B_{P^c}$ is a basis for $M(W)/F(P)$. Furthermore, for any basis, $B$, of $W$ such that $|B \cap F(P)| = |P|$, the intersection $B \cap V_{P^c}$ is an independent subset of $V_{P^c}$ of size $|P^c|$. Since $B \cap F(P)$ is also a basis of $W|P$, the set $B \cap V_{P^c}$ is an independent subset of every subdiagram of $W$, $W'$, that contains $W/P$ as a subdiagram. Therefore, again by Proposition \ref{indepincreasecor}, $B\cap V_{P^c}$ is a basis of $W/P$, and the equality holds.

Now suppose $\rk F(P) < |P|$. Then, for every basis $B$ of $W$ the set $B \setminus B_{P^c} \not \subset F(P)$. This is because $B_{P^c}$ is a basis of $W/P$. Therefore $|B \setminus B_{P^c}| = |P|$, which is greater than the rank of $F(P)$. Write the bases of $M(W)/F(P)$ as $B \setminus F(P)$. We see that for every such basis, there is a $B_{P_c}$ such that $B_{P_c} \subsetneq B \setminus F(P)$. Therefore, the equality does not hold.
\end{proof}

Since the Wilson loop $W/P$ is a Wilson loop in its own right, we have the following result:

\begin{cor}
The matroid $M(W)/F(P)$ is disconnected if and only if there is a partition of $P^c$, $Q_1$, $Q_2$, such that $V_{Q_1}$ and $V_{Q_2}$ partition $V_Q$. In particular $V_{Q_1} \cap V_{Q_2} = \emptyset$.
\end{cor}

\begin{proof}
This follows directly from definition \ref{loopcontractdfn}, Theorem \ref{contractionmatroid} and Theorem \ref{disconnectedloopprop}.
\end{proof}

We summarize the results of this and the previous section as follows:

\begin{thm}
Let $W= (\fP,n)$ be a connected Wilson loop. A flacet of $M(W)$ is defined by a set of propagators, $P \subset \fP$, satisfying the following constions: \begin{enumerate}\item The corresponding propagator flat, $F(P)$, is a cyclic flat. In particular, it has maximal rank. \item For every partition of propagators, $Q_1 \amalg Q_2 = P^c$, there is not a partition of dependency sets, $V_{Q_1} \amalg V_{Q_2} \neq  V_{P_c}$. \item For every partition of propagators, $P_1 \amalg P_2 = P$, there is not a partition of propagator flats, $F(P_1) \amalg F(P_2) \neq  F(P)$.
\end{enumerate} \label{propagatorflacetthm}
\end{thm}

These subsets of propagators are important enough that they deserve a name.

\begin{dfn}
Given a Wilson loop diagram $(\fP, [n])$, a subset $P \subset \fP$ is a propagator flacet if and only if it defines a flacet of $M(W)$ as in Theorem \ref{propagatorflacetthm}.
\end{dfn}

\subsubsection{Postive Wilson loops\label{identify positroids}}
In Theorem \ref{propagatorflacetthm}, we identified the propagators of a Wilson loop diagram that correspond to flacets. In this section, we use this to characterize  which Wilson loop diagrams lead to positroids. This is a direct consequence of Theorem \ref{positroidcondition}.

\begin{thm} \label{graphpositroidcondition}
A Wilson loop diagram defines a positroid if and only if all propagator flacets, $P$, define flats, $F(P)$, that are cyclic intervals.
\end{thm}

\begin{proof} Follows from Theorem \ref{positroidcondition}.\end{proof}

We now apply this to the diagramatics of Wilson loop diagrams to characterise propagator configurations that lead to positroids.

\begin{thm}
Consider the connected Wilson loop diagram $W = (\fP, [n])$. If the matroid $M(W)$ is not a positroid, then $W$ has crossing propagators.
\end{thm}

\begin{proof}
By Theorem \ref{graphpositroidcondition}, if $M(W)$ is not a positroid, then there exists a propagator flacet $P$ that defines a cyclic flat, $F(P)$, which is not a cyclic interval. Write $F(P) = \cup_1^l [S_{2j}]$, where each $[S_j]$ is a cyclic interval in $F(P)$. Similarly, write $F(P)^c = \cup_1^l [T_{2j-1}]$, such that $[T_k]$ immediately preceeds $[S_{k+1}]$ in $[n]$.

Write $Q = \fP \setminus P$. Then $W/P = W(Q, V_Q)$. 
Since $M(W/P)$ is connected, by Theorem \ref{disconnectedloopprop}, there does not exist a partition $V_1$, $V_2$ of $V_Q$ such that $\Prop(V) \cap \Prop(V_2) = \emptyset$.

Similarly, $M(W)|F(P)$ is connected. There does not exist a partition $F_1$, $F_2$ of $F(P)$ such that $\Prop(F_1) \cap \Prop(F_2) = \emptyset$. In particular, there exists a propagator, $p \in P$, with endpoints in distinct cyclic intervals, $V_{p} \cap S_{2l}, V_{p} \cap S_{2m} \neq \emptyset$, for $l \neq m$. Then $p$ divides $V_Q$ into two sets, $V_{lm} = \cup_{l < j \leq m} [T_{2j-1}]$ and $U_{ml}\cup_{m < j \leq l} [T_{2j-1}]$. Here, as usual, the ordering is given cyclically. Since $W/P$ is connected, there is a propagator $q$ such that $V_{q} \cap T_{2j-1}, V_{p} \cap T_{2k-1} \neq \emptyset$, with $l < j \leq m$ and $m < k \leq l$. Therefore, by definition \ref{crossingpropagatordef}, the propagators $q$ and $p_{lm}$ cross in $W$.
\end{proof}

The contrapositive statement of this theorem gives an important statement about the relationship between Wilson loop diagrams and positroids.

\begin{cor}
Any Wilson loop with non-crossing propagators defines a positroid. \label{noncrosspos}\end{cor}

However, this is not a full classification of which Wilson loop diagrams lead to positroids. In particular, there are Wilson loop diagrams with crossing propagators, that lead to positroids. We give an example of this below.

\begin{eg}
Consider the Wilon loop diagram with crossing propagators,

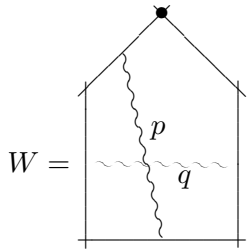
\bas W = {\begin{xy} (0,20) *{\bullet},
(-10, 11); (-10, -11) **{\dir{-}},
(-11, -10); (11, -10) **{\dir{-}},
(10, -11); (10, 11) **{\dir{-}},
(11,9); (-1, 21) **{\dir{-}},
(-11,9); (1, 21) **{\dir{-}},
(-5, 15); (0, -10) **@{~}?/0pt/+(2,2)*{p},
(-10, 0); (10, 0) **@{~}?/0pt/+(3,-2)*{q} ,
\end{xy}} \eas

The relevant flats are $F(p) = \{1\}$, $F(q) = \{5\}$, and $F(p,q) = [5]$. The flats $F(p)$ and $F(q)$ are trivial flats, and therefore need not be considered. $F(p,q)$ is trivially a flacet, and a cyclic interval. Therefore, $M(W)$ is a positroid. We may also see this on the level of the matrices associatated to Wilson loops. Recall that
\bas \cM(W(\cZ_*)) = \left(
\begin{array}{ccccc}
c_{1,1} & c_{1,2} & c_{1,3} & c_{1,4} & 0 \\
0 & c_{2,2} & c_{2,3} & c_{2,4} & c_{2,5} \\
\end{array}
\right) \;.\eas This matrix has non-negative minors if $\frac{c_{1,2}}{c_{1,3}} \geq \frac{c_{2,2}}{c_{2,3}}$ and $\frac{c_{1,3}}{c_{1,4}} \geq \frac{c_{2,3}}{c_{2,4}}$, and $c_{1,1}, c_{2,5} > 0$.
\label{positivecrossingeg}\end{eg}

In general, given a Wilson loop diagram with crossing propagators, if it is not a positroid, one may always add propagators until all flats are trivial. For instance, $
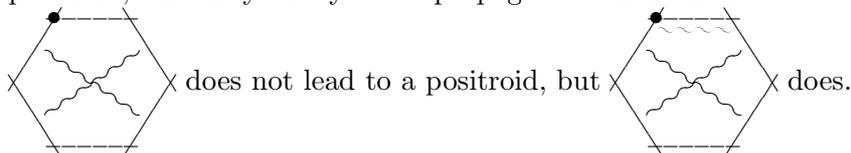
\begin{xy} (-5,8.5) *{\bullet},
(-6, 8.5); (6, 8.5) **{\dir{-}},
(-4, 10); (-11, -1) **{\dir{-}},
(4, 10); (11, -1) **{\dir{-}},
(4,-10); (11, 1) **{\dir{-}},
(-4,-10); (-11, 1) **{\dir{-}},
(-6, -8.5); (6, -8.5) **{\dir{-}},
(-7, -5); (7, 5) **@{~},
(-7, 5); (7, -5) **@{~},
\end{xy}$ does not lead to a positroid, but $\begin{xy} (-5,8.5) *{\bullet},
(-6, 8.5); (6, 8.5) **{\dir{-}},
(-4, 10); (-11, -1) **{\dir{-}},
(4, 10); (11, -1) **{\dir{-}},
(4,-10); (11, 1) **{\dir{-}},
(-4,-10); (-11, 1) **{\dir{-}},
(-6, -8.5); (6, -8.5) **{\dir{-}},
(-7, -5); (7, 5) **@{~},
(-7, 5); (7, -5) **@{~},
(-6, 7); (6, 7) **@{~},
\end{xy}$ does.

Both this last example, and the Wilson loop diagram in Example \ref{positivecrossingeg} are exact Wilson loops, as defined in Definition \ref{overexact}. We have seen from theorem \ref{equivalencethm}, that if two exact Wilson loops on $n$ points are equivalent, as defined in definition \ref{equivdiagramdef}, they define the same subspace of $\R^n$. In other words, they define the same matroid. This gives the following result:

\begin{thm}
Let $W = (\fP, [n])$ be a Wilson loop diagram such that every disjoint set of crossing propagators $R_i$ can be extended to a set $P_i \subset \fP$ such that $|V_{P_i}| = |P_i| +3$. Then $W$ defines a positroid. \label{crosspos}
\end{thm}

\begin{proof}
Given any Wilson loop as above, define a Wilson loop diagram $W'$ that replaces each $P_i$ with a set of non-crossing propagators $Q_i$ such that $V_{P_i} = V_{Q_i}$. Then, by definition \ref{equivdiagramdef}, $W$ and $W'$ are equivalent. Therefore, they define the same subspace of $\R^n$, and the same matroid. Since $W'$ is constructed to not have any crossing propagators, $W'$ is a positroid, by Corrollary \ref{noncrosspos}. Therefore, $W$ is a positroid.
\end{proof}

Therefore, we have shown several things. The object of physical interest is the Wilson loop amplitude, $A_{n, k} (\cZ)$, at $N^kMHV$. This is defined by the sum of all Wilson loop integrals at $N^kMHV$. For large $k$ and $n$, it is important to be able to break the integral apart into independent components. Theorem \ref{disconnectedloopprop} does just that, by identifying the connected components of an integral. Theorem \ref{overdefinethm} shows that overdefined Wilson loop diagrams need not be considered in this sum at all. Corrollary \ref{noncrosspos} and Theorem \ref{noncrosspartconn} show that any well defined diagram with non-crossing propagators must be considered. Finally, Theorem \ref{crosspos} and Theorem \ref{equivalencethm}, raise an interesting question about the integrals associated to exact Wilson loop diagrams.

\begin{rmk} Exact Wilson loop diagrams are precisely those that have $\overline{MHV}$ subdiagrams. That is, if $(\fP, [n])$ is an exact, not necessarily connected Wilson loop diagram, and $P \subset \fP$ is a subset of propagators satisfying \eqref{exactcond}, then, by definition, the Wilson loop sub diagram $(P, V_P)$ is an $\overline{MHV}$ diagram (i.e. one in which $k + 3 = n)$, and therefore physically uninteresting. However, the original diagram $(\fP, [n])$ need not be of $\overline{MHV}$ type. Furthermore, in the case of exact diagrams, two planar Wilson loop diagrams with the same initial set of twistor data define the same Grassmannians. However, one sees from the from of the $\hat c_{p,r}$ defining the integrals associated to these diagrams, equation \eqref{amplitudeeq}, two different planar diagrams, given the same generic twistor configuration, give rise to different integrals. This is not a problem that occurs in the case of non-exact diagrams, where, conjecturally, each diagram defines a different Grassmannian. In this case, the integral of a Wilson loop Diagram uniquely assigns a function of the twistor configuration to the Grassmannian associated to it. The combinatorics and the physical characteristics of exact Wilson loop diagrams deserve further study. \label{amplitudeambiguity}\end{rmk}

\section{Future work}\label{futurework}

This is a first step in a project that aims to understand the amplituhedron and its extension to correlators, the `correlahedron', via MHV diagrams. Ultimately these are a stepping stone to the dual amplituhedron \cite{hodges:20013eliminating}, the polytope describing amplitudes and loop integrands etc., by regarding each diagram as providing some cell in some invariantly defined polyhedron whose volume provides the amplitude, correlator or their integrands.

This paper hopes to encourage the application of a hitherto untapped set of tools to this problem encoding these physically intersting objects in a new language, that of matroids.

For this programme, the positive Grassmannian is an intermediate stage that provides a better understanding of the cells that each diagram gives rise to in a form where they can be mapped into the ambient space of the amplituhedron. Key ideas in this programme are as follows.

\begin{enumerate}
\item The rational functions that diagrams provide have both physical poles corresponding to physical propagators in conventional Feynman diagrams and poles that depend on the reference twistor, and hence are spurious and unphysical.  Both types of poles correspond to where the $\hat c_{p,s}$'s change sign at $0$ and $\infty$ (on the support of the delta functions, the $\hat c_{p,s}$s can be expressed as functions of  the twistor input data).  These in turn are  the boundaries of the positive regions and hence of the cells.
\item In a physical object such as a full amplitude, the spurious poles must all cancel when all the diagrams are summed, leaving only the physical singularities.  Seeing this cancellation algebraically is challenging except in simple cases.  However, the amplituhedron turns the problem of understanding the cancellation of spurious poles into the geometric problem of seeing that the cells or tiles can be joined together across their spurious boundaries leaving only boundaries corresponding to the physical singularities.
\item As mentioned in Section \ref{matroiddefssection}, each matroid defines a family of Grassmannians. There is a lot known about the submanifolds of Grassmannians corresponding to matroids \cite{postroidmanifold13}. In the Wilson Loop Diagram context, one expects each Wilson loop diagram to define a cycle, or a closed submanifold, of the positive Grassmanians. Given the natural matroidal interpretation of Wilson loop diagrams, one hopes that the geometry of positroids will help in understanding what the boundaries of these cycles are, and how they embed into the geometry of the positive Grassmannians.
\item One problem that we have identified is that additional data is needed over and above the positroid data of a diagram to determines its precise contribution when there are subdiagrams with $n=k+3$, $k\geq 2$.  Although the formulae are clear, the optimal way to encode this additional data needs to be found.
\item For performing concrete calculations on the matrices defined by the Wilson loop diagrams, one often wishes to know how the relative signs of the coefficients $\hat c_{p,x}$ depend on the Wilson loop diagram. This is well studied in the matroid community using tools like Le diagrams, see, for example, \cite{Postetal09}.

\end{enumerate}

\subsection{Acknowledgements:}
We would like to thank Lionel Mason, Alex Fink, Paul Heslop, Karel Casteel, and Reza Doobary for many useful discussions. This research has been made possible by support from the EPSRC under EP/J019518/1 and the ERC under DUALITIESHEPTH.

\bibliographystyle{amsplain}
\bibliography{c:/users/S/Dropbox/bibliography/Bibliography}{}

\end{document}